%% file: inchworm_info_reuse.tex
\documentclass[11pt, reqno]{amsart}
\usepackage{amsfonts}
\usepackage{amssymb,physics}
\usepackage{amsmath}
\usepackage{tikz}
\usetikzlibrary{patterns,snakes,arrows.meta}
\usepackage{algorithm}
\usepackage{algpseudocode} 
\usepackage{scalerel}
\usepackage{subfigure}
\usepackage{graphicx}

\allowdisplaybreaks[4]

\def\Ls{\mathcal{L}}
\def\Is{\mathcal{I}}
\def\Us{\mathcal{U}}

\def\Ms{\mathcal{M}}

 \def\Ss{\mathcal{S}}
\def\sgn{\mathrm{sgn}}
\def\sb{\boldsymbol{s}}
\def\dd{\mathrm{d}}
\def\ii{\mathrm{i}}
\def\b1{\boldsymbol{1}}
\def\sf{s_{\mathrm{f}}}
\def\si{s_{\mathrm{i}}}

\def\sT{\hat{T}}
\def\st{\tilde{s}}

\def\ssb{{\hat{\sb}}}
\def\P{\mathbb{P}}
\def\tr{\mathrm{tr}}
\def\ee{\mathrm{e}}
\def\ii{\mathrm{i}}

\newcommand\mQ{\mathcal{Q}}
\newcommand{\mf}[1]{\mathfrak{#1}}
\newcommand{\mc}[1]{\mathcal{#1}}

\newcommand{\redtext}[1]{{\color{red}#1}}

\definecolor{darkgreen}{rgb}{0.0, 0.65, 0.29}

\newtheorem{proposition}{Proposition}

\newtheorem{lemma}{Lemma}
\newtheorem{theorem}{Theorem}
\theoremstyle{remark}
\newtheorem{remark}[theorem]{Remark}

\title[Fast algorithms of bath calculations]{Fast algorithms of bath calculations in simulations of quantum quantum system-bath dynamics}

\author{Zhenning Cai}
\address[ZC]{Department of Mathematics, National University of
  Singapore, Level 4, Block S17, 10 Lower Kent Ridge Road, Singapore 119076.}
\email{matcz@nus.edu.sg}

\author{Jianfeng Lu}
\address[JL]{Department of Mathematics, Department of Physics, and
  Department of Chemistry \\ Duke University, Box 90320, Durham NC 27708, USA.}
\email{jianfeng@math.duke.edu}

\author{Siyao Yang}
\address[SY]{Department of Mathematics, National University of
  Singapore, Level 4, Block S17, 10 Lower Kent Ridge Road, Singapore 119076.}
\email{matsiya@nus.edu.sg}

\thanks{Zhenning Cai's work was supported by the Academic Research Fund of the Ministry of Education of Singapore under grant R-146-000-291-114. The work of JL~was supported in part by the National Science Foundation via grants DMS-2012286 and CHE-2037263.}

\begin{document}

\begin{abstract}
    We present fast algorithms for the summation of Dyson series and the inchworm Monte Carlo method for quantum systems that are coupled with harmonic baths. The algorithms are based on evolving the integro-differential equations where the most expensive part comes from the computation of bath influence functionals. To accelerate the computation, we design fast algorithms based on reusing the bath influence functionals computed in the previous time steps to reduce the number of calculations. It is proven that the proposed fast algorithms reduce the number of such calculations by a factor of $O(N)$, where $N$ is the total number of time steps. Numerical experiments are carried out to show the efficiency of the method and to verify the theoretical results.
\end{abstract}

\date{\today}
\keywords{Dyson series; inchworm Monte Carlo method; integro-differential equation; accelerated bath calculation; fast algorithms}

\maketitle

\section{Introduction}

In classical thermodynamics, many processes are irreversible due to the dissipation of energy. To describe such an effect at the quantum level, quantum dissipation has been widely studied in the literature, and one of the successful approaches is the Caldeira-Leggett model \cite{Caldeira1983a,Caldeira1983b}, which assumes that the quantum system is coupled with a harmonic bath. The presence of the bath leads to non-Markovian and irreversible dynamics of the quantum system. The system-bath dynamics has also been extensively used to study quantum decoherence, which leads to classical behavior of the quantum systems. In addition to its theoretical importance, the model is widely used to describe interaction of a quantum system with its environment, and has applications in a number of fields  including quantum optics \cite{Breuer2007}, quantum computation \cite{Nielsen2010}, and dynamical mean field theory \cite{Gull2011RMP}. 

The main challenge for simulating Caldeira-Leggett type models lies in the huge degrees of freedom associated with the harmonic bath, which makes the direct calculation of the wave function impossible in practice. For decades, many techniques for dimension reduction have been developed in order to avoid solving the harmonic bath directly. Some classical numerical methods based on path integrals, such as the quasi-adiabatic propagator path integral (QuAPI) method \cite{Makri1995, Makri1996}, the iterative QuAPI-based methods \cite{Makarov1994, Makri2017} and the hierarchical equations of motion (HEOM) \cite{Strumpfer2012}, introduce the bath effects using the influence functional \cite{Feynman1963} and can produce numerically exact results, while a considerably large memory cost is often required. A wave function-based approach known as the multiconfiguration time-dependent Hartree (MCTDH) method \cite{Beck2000,Meyer1990}, as well as its multilayer formulation (ML-MCTDH) \cite{Wang2003}, has achieved impressive success in molecular systems, although they may become harder to converge for the nonequilibrium heat transport in the Caldeira-Leggett model \cite{Chen2017b}.

Another conventional approach to the system-bath dynamics is the generalized quantum master equation (GQME) \cite{Zwanzig1960,Nakajima1958,Mori1965} obtained by applying the Nakajima-Zwanzig projection operator, which reduces the dissipative bath term to a memory kernel. Such formulation provides an exact integro-differential equation for simulating the reduced dynamics. However, the evaluation of the memory kernel could be challenging due to its dependence on the projector. To alleviate this difficulty, \cite{Shi2003,Zhang2006} have proposed new approaches to calculating the memory kernel based on its projection-free formulations. The transfer tensor method (TTM) \cite{Cerrillo2014} based on the discretization of GQME is later introduced, which is also applied \cite{Rosenbach2016} in an method called the time evolving
density matrix using orthogonal polynomials algorithm (TEDOPA) \cite{Chin2010,Prior2010} to reduce the size of the propagator. Further development on the evaluation of memory kernel includes \cite{Kidon2018} where the memory kernel is related to the evolution of a reduced system propagator which is numerically computed by ML-MCTDH, and \cite{Kelly2013} which computes the memory kernel based on semiclassical trajectories.

An alternative to these deterministic approaches gaining popularity in the recent years is a class of stochastic methods known as the diagrammatic quantum Monte Carlo (dQMC) \cite{Prokof'ev1998,Werner2009}, which have been shown to be powerful in describing the equilibrium physics of impurity models. The underlying idea is to replace the expensive high-dimensional integrals in the Dyson series of the quantum observable by the average of unbiased samples of diagrammatic expansions \cite{Chen2017,Muhlbacher2008,Werner2006}. For example in GQME, the memory kernel can be evaluated stochastically using the real time path integral Monte Carlo \cite{Cohen2011,Cohen2013}. However, such an approach severely suffers from the notorious numerical sign problem \cite{Chen2017,Cai2020,Cai2020b}, meaning that the variance of the numerical solution grows at least exponentially with time. To maintain the accuracy of the results, a large number of Monte Carlo samples need to be drawn as time increases, leading to an extremely expensive computational cost on the evaluation of the bath influence functional. To mitigate the sign problem, many techniques such as stochastic unraveling of influence functionals \cite{StockBurger2002} and multilevel blocking Monte Carlo \cite{Mak1992,Egger2000,Muhlbacher2003} have emerged throughout the past several decades. Recently, the inchworm Monte Carlo method \cite{Cohen2015,Chen2017} based on the partial resummation of Dyson series has been proposed, which has proven impressive capability to relieve the sign problem both numerically \cite{Chen2017b,Cai2020} and theoretically \cite{Cai2020b}. Nevertheless, the computations of the bath influence functional remain to be the major 
bottleneck \cite{Boag2018,Yang2021} even after such reductions. In this paper, we consider a strategy to further reduce the cost of bath calculations in the summation of Dyson series and inchworm Monte Carlo method.  

The central idea to reduce bath calculations lies in the invariance of the influence functional in Dyson series or inchworm method, which is formulated as a summation over some pairwise bath interactions. In detail, the expected value of an observable $O$ can be written as $\tr \bigl(\rho(0)\ee^{\ii t H} O \ee^{-\ii t H}\bigr)$ with $\rho(0)$ being the initial density matrix and $H$ the quantum Hamiltonian. In the diagrammatic Monte Carlo methods, such an expression is often denoted using the unfolded Keldysh contour \cite{Keldysh1965} plotted in Figure \ref{fig:keldysh}. By Wick's theorem, computing the trace requires us to evaluate the correlation function $B(\tau_1,\tau_2)$ of two time points $-t \le \tau_1 \le \tau_2 \le t$, which can be diagrammatically represented as an arc in Figure \ref{fig:keldysh}. This two-point correlation function satisfies the translational invariance (when $\tau_1$ and $\tau_2$ are on the same side of the origin) and the stretching invariance (when $\tau_1$ and $\tau_2$ are on different sides of the origin). Making use of this property can greatly reduce the computational cost for the bath calculation.

Let us remark that the invariance of the two-point correlation function is also utilized in the recently proposed SMatPI (small matrix decomposition of the path integral) method \cite{makri2020}, which is an improved version of the iterative QuAPI method. In SMatPI, the bath integrand factor is computed using the Feynman-Vernon influence functional. The SMatPI method groups a number of paths into some small matrices, and using the translational invariance of the Feymann-Vernon influence functional, the information in the small matrices can be directly used in future time steps without being recalculated. In our method, the bath influence functional is the sum of a lot of diagrams, and the reuse of previously calculated functionals avoids recomputation of all the translated or stretched diagrams included, which significantly enhances the computational efficiency.

\input{images/fig_keldysh}

The rest of this paper is organized as follows. In Section \ref{sec:dyson reuse}, we introduce the spin-boson model and its Dyson series expansion. An integro-differential equation associated with Dyson series is then derived, based on which we propose a fast algorithm where the previous bath calculations are reused. An analysis on computational cost is included to examine the performance of the proposed algorithm. Such framework is then applied to Section \ref{sec:inchworm reuse} where the more complicated inchworm Monte Carlo method is studied. Some numerical experiments are carried out in Section \ref{sec:num exp} to verify the theoretical results in Section \ref{sec:dyson reuse} and \ref{sec:inchworm reuse}, and test the order of convergence of the fast algorithms. Finally, some conclusions and discussions are given in Section \ref{sec:conclusion}.

\section{Fast calculation of time evolution of Dyson series}
  \label{sec:dyson reuse}
\subsection{Introduction to spin-boson model and Dyson series}
\label{sec:dyson}
We study the system-bath dynamics described by the von Neumann equation for the density matrix $\rho(t)$   
\begin{equation} \label{eq:vonNeumann}
\ii \frac{\dd \rho}{\dd t} = [H, \rho]:= H\rho - \rho H,
\end{equation}
where the Schr\"odinger picture Hamiltonian $H$ is a Hermitian operator on the Hilbert space $\mc{H} = \mc{H}_s \otimes \mc{H}_b$, with $\mc{H}_s$ and $\mc{H}_b$ representing respectively the Hilbert spaces associated with the system and the bath. The Hamiltonian $H$ consists of the Hamiltonians of the system and the bath, as well as a coupling term describing the interaction of the system and the bath. Assuming that the coupling term has the tensor-product form, we have
\begin{displaymath}
 H =  H_s \otimes \mathrm{Id}_b + \mathrm{Id}_s \otimes H_b +  W_s \otimes W_b,
\end{displaymath}
where $H_s,W_s \in \mc{H}_s$, $H_b,W_b \in \mc{H}_b$, and $\mathrm{Id}_s,\mathrm{Id}_b$ are the identity operators for the system and the bath, respectively. In our paper, we take the common assumption that the bath is modeled by a larger number of harmonic oscillators. While the algorithms discussed in this work can be easily generalized to any multiple-state open quantum systems, we only consider the simplest system modeled by a single spin. Such a problem contains most difficulties in the treatment of the system-bath coupling, which is known as the spin-boson model to be introduced below.

\subsubsection{Spin-boson model}
As one fundamental example of the system-bath dynamics \cite{Wang2000, Kernan2002, Duan2017}, the spin-boson model assumes that
\begin{displaymath}
\mc{H}_s = \text{span}\{ \ket{0}, \ket{1} \}, \qquad
\mc{H}_b = \bigotimes_{l=1}^L \left( L^2(\mathbb{R}^3) \right),
\end{displaymath}
where $L$ is the number of harmonic oscillators in the bath. The corresponding Hamiltonians are
\begin{displaymath}
H_s = \epsilon \hat{\sigma}_z + \Delta \hat{\sigma}_x, \qquad
H_b = \sum_{l=1}^L \frac{1}{2} (\hat{p}_l^2 + \omega_l^2 \hat{q}_l^2)
\end{displaymath}
Here $\hat{\sigma}_x$, $\hat{\sigma}_z$ are Pauli matrices satisfying
  $\hat{\sigma}_x \ket{0} = \ket{1}$, $\hat{\sigma}_x \ket{1} = \ket{0}$,
  $\hat{\sigma}_z \ket{0} = \ket{0}$, $\hat{\sigma}_z \ket{1} = -\ket{1}$, and the parameters $\epsilon$, $\Delta$ are respectively the energy difference between two spin states and the frequency of the spin flipping. In the bath Hamiltonian $H_b$, the notations $\hat{p}_l$, $\hat{q}_l$ and $\omega_l$ are respectively the momentum operator, the position operator and the frequency of the $l$th harmonic oscillator. The coupling operators are given by  
\begin{displaymath}
 W_s = \hat{\sigma}_z ,\qquad W_b = \sum_{l=1}^L c_l \hat{q}_l,
\end{displaymath}
where $c_l$ is the coupling intensity between the $l$th harmonic oscillator and
the spin.

The density matrix solving \eqref{eq:vonNeumann} can be written as $\rho(t) =  \ee^{-\ii t H} \rho(0) \ee^{\ii t H}$, and we assume its initial value has the separable
form $\rho(0) = \rho_s \otimes \rho_b$ with the initial bath $\rho_b$ being the thermal equilibrium $\exp(-\beta H_b)$, where $\beta$ is the inverse temperature. We are interested in the evolution of the expectation for a given observable $O = O_s \otimes \mathrm{Id}_b$ acting only on the system, defined by
\begin{equation} \label{eq:O(t)}
\langle O(t) \rangle := \tr(O \rho(t))
  = \tr(O \ee^{-\ii t H} \rho(0) \ee^{\ii t H})=  \tr(\rho_s \otimes \rho_b \ee^{\ii t H} O_s \ee^{-\ii t H} ) = \tr_s (\rho_s G(-t,t))
\end{equation}
where the propagator $G(-t,t):=\tr_b(\rho_b \ee^{\ii t H} O_s \ee^{-\ii t H} )$ is a $2\times 2$ Hermitian matrix due to the cyclic property of the trace operator:
\begin{equation}\label{G hermitian}
 G(-t,t)^\dagger = \tr_b(\ee^{\ii t H^\dagger} O_s^\dagger \ee^{-\ii t H^\dagger} \rho_b^\dagger )  =  \tr_b(\ee^{\ii t H} O_s \ee^{-\ii t H} \rho_b ) = \tr_b(\rho_b \ee^{\ii t H} O_s \ee^{-\ii t H}  ) = G(-t,t).
\end{equation}

\subsubsection{Dyson series}
Due to the high dimensionality of the space $\mc{H}_b$, it is impractical to solve $\ee^{\pm \ii t H}$ directly. One feasible approach is to apply the method of quantum Monte Carlo to approximate $G(-t,t)$ numerically. It is well known that $G(-t,t)$ can be expanded into the following \emph{Dyson series} (for derivation, see \cite{Cai2020}):
\begin{equation}\label{G -t t}
    G(-t, t)  =  \ee^{\ii t H_s} O_s \ee^{- \ii  t H_s} + \sum_{m=1 }^{+\infty}
  \ii^m \int_{-t \le \sb \le t} \dd \sb 
(-1)^{\#\{\sb < 0\}}  \mathcal{U}^{(0)}(-t, \sb , t) \cdot
    \mathcal{L}_b(\sb) , \text{~for~} t \geq 0.
 \end{equation}
The above formula is interpreted as:
\begin{itemize}
\item Integral notation: for any $a\le A$
\begin{displaymath}
\int_{a \le \sb \le A} \dd \sb  : = \int^{A}_{a} \dd s_m \int^{s_m}_{a} \dd s_{m-1} \cdots \int^{s_{2}}_{a} \dd s_1.
\end{displaymath}
\item $\#\{\sb < 0\}$: number of components in $\sb=(s_1,s_2,\cdots,s_m)$  that are less than $0$.
\item System associated functional $\mathcal{U}^{(0)}$:
\begin{equation}\label{def U0}
 \mathcal{U}^{(0)}(-t,\sb,t)  = G_s^{(0)}(s_m, t) W_s G_s^{(0)}(s_{m-1}, s_{m}) W_s
  \cdots W_s G_s^{(0)}(s_1, s_2) W_s G_s^{(0)}(-t, s_1),
\end{equation}
where
\begin{equation}\label{def:Gs}
   G_s^{(0)}(\si, \sf) =
  \begin{cases}
    \ee^{-\ii (\sf - \si) H_s},
    & \text{if } \si \le \sf < 0, \\[5pt]
    \ee^{-\ii (\si - \sf) H_s},
    & \text{if } 0 \le \si \le \sf, \\[5pt]
    \ee^{\ii  \sf H_s} O_s \ee^{\ii  \si H_s},
    & \text{if } \si < 0 \le \sf.
  \end{cases}
\end{equation}

\item Bath influence functional $\Ls_b$:
\begin{equation} \label{eq:L all pair}
   \Ls_b(s_1,\cdots,s_m) =  \left\{   \begin{array}{l l}
   0, & \text{if $m$ is odd}; \\ 
   \displaystyle \sum_{\mf{q} \in \mQ(\sb)} \prod_{(s_j,s_k) \in \mf{q}} B(s_j,s_k), &  \text{if $m$ is even},
    \end{array} \right.
\end{equation}
where $B: \{(\tau_1,\tau_2) \mid  \tau_1 \leq \tau_2 \} \rightarrow \mathbb{C}$ is the \emph{two-point bath correlation} whose value only relies on the difference of the absolute values of the two variables:
\begin{equation}\label{def:B}
B(\tau_1, \tau_2) = B^*(\Delta \tau) = \frac{1}{\pi} \int^{\infty}_0 J(\omega) \left[  \coth\left(\frac{\beta \omega }{2}\right) \cos(\omega \Delta \tau) - \ii \sin(\omega \Delta \tau)  \right] \dd \omega
\end{equation}
with 
\begin{equation*}
 \Delta \tau = |\tau_1| - |\tau_2|.
\end{equation*}
The explicit formula of the single-variable function $B^*(\cdot)$ depends on the real-valued spectral density $J(\omega)$. The set $\mQ(\sb)$ is given by:
\begin{equation} \label{eq:all linking pairs}
  \begin{split}
  & \mQ(s_1,\cdots,s_m) =\\
  &\Big\{ \{(s_{j_1}, s_{k_1}), \cdots, (s_{j_{m/2}}, s_{k_{m/2}})\} \,\Big\vert\,  \{j_1, \cdots, j_{m/2}, k_1, \cdots, k_{m/2}\} = \{1,\cdots,m\}, \\
  & \hspace{120pt} j_l < k_l \text{ for any } l = 1,\cdots,m/2
  \Big\}.
  \end{split}
\end{equation}  
 
\end{itemize} 

To get some intuition behind the definition of the bath influence functional, we consider a simple case $m=4$, where the equation \eqref{eq:L all pair} turns out to be
\begin{equation} \label{eq:all pairs example}
   \Ls_b(s_1,s_2,s_3,s_4) = B(s_1,s_2) B(s_3,s_4) + B(s_1,s_3) B(s_2,s_4) + B(s_1,s_4) B(s_2,s_3),
\end{equation}
which can be graphically represented by the following diagrams:
\begin{equation} \label{eq:all linking pair diagram example}
\Ls_b(s_1,s_2,s_3,s_4)
=
\begin{tikzpicture}
\draw[-] (0,0)--(1.5,0);\draw plot[only marks,mark =*, mark options={color=black, scale=0.5}]coordinates {(0,0) (0.5,0) (1,0)(1.5,0)};
\draw[-] (0,0) to[bend left=75] (0.5,0);
\draw[-] (1,0) to[bend left=75] (1.5,0);
 \end{tikzpicture}
 +
 \begin{tikzpicture}
\draw[-] (0,0)--(1.5,0);\draw plot[only marks,mark =*, mark options={color=black, scale=0.5}]coordinates {(0,0) (0.5,0) (1,0)(1.5,0)};
\draw[-] (0,0) to[bend left=75] (1,0);
\draw[-] (0.5,0) to[bend left=75] (1.5,0);
 \end{tikzpicture}
 +
  \begin{tikzpicture}
\draw[-] (0,0)--(1.5,0);\draw plot[only marks,mark =*, mark options={color=black, scale=0.5}]coordinates {(0,0) (0.5,0) (1,0)(1.5,0)};
\draw[-] (0,0) to[bend left=60] (1.5,0);
\draw[-] (0.5,0) to[bend left=75] (1,0);
 \end{tikzpicture}~.
\end{equation}
In the diagrammatic representation above, each diagram refers to a product $B(\cdot,\cdot)B(\cdot,\cdot)$ where each arc connecting a pair of bullets denotes the corresponding two-point correlation. For general $m$, the value of the corresponding bath influence functional is the sum of all possible combinations of such pairings, and the number of these diagrams is $(m-1)!!$. Since the bath influence functional vanishes when $m$ is odd, the right-hand side of \eqref{G -t t} actually only sums over terms with even $m$.

To evaluate $G(-t,t)$, one may truncate the Dyson series at a sufficiently large even integer $\bar{M}$ and evaluate those high-dimensional integrals on the right-hand side using Monte Carlo integration, resulting in the bare dQMC. More specifically, one can draw $\Ms$ samples of $\{m^{(i)}, \sb^{(i)}\}$ independently according to a certain distribution $\P(m,\sb)$ for $m=2,\cdots,\bar{M}$ and $\sb$ satisfying $-t\le s_1\le \cdots \le s_m \le t$. Then $G(-t, t)$ can be approximated by  
\begin{equation}\label{bare dqmc}
  G(-t, t)  \approx  \ee^{\ii t H_s} O_s \ee^{- \ii  t H_s} + \frac{1}{\Ms} \  \sum^{\Ms}_{i = 1} \  \frac{1}{\P(m^{(i)},\sb^{(i)})} \cdot    \ii^{m^{(i)}} 
(-1)^{\#\{\sb^{(i)} < 0\}}  \mathcal{U}^{(0)}(-t, \sb^{(i)} , t) 
    \mathcal{L}_b(\sb^{(i)}). 
\end{equation}
The numerical solution obtained via bare dQMC has been proved to have a variance that grows double exponentially with respect to $t$ \cite{Cai2020b}. Therefore, the number of samples $\Ms$ should increase with $t$ accordingly to achieve sufficient accuracy at the final time. Hence, the computational cost of the Monte Carlo approximation, especially the expensive evaluation of the bath influence functional $\Ls_b$, also grows double exponentially with time. To mitigate this problem, in the next section, we will formulate an integro-differential equation which gives the time evolution of $G(-t,t)$. Thus some bath influence functionals obtained when computing $G(-t', t')$ with $t' < t$ can be reused when computing $G(-t,t)$.
Before that, however, we first present the following useful properties of the bivariate functions $G^{(0)}_s(\cdot,\cdot)$ and $B(\cdot,\cdot)$ appearing in the definitions of $\mc{U}^{(0)}$ and $\Ls_b$:

\begin{proposition}
\
\begin{itemize}
 \item For any $\si\le \sf$, we have 
 \begin{equation}\label{prop: B and Gs}
\begin{split}
&G_s^{(0)}(-\sf,-\si) = G_s^{(0)}(\si,\sf)^\dagger \text{~for $\si \neq 0$ and $\sf \neq 0$}, \\
&B(-\sf,-\si) = \overline{B(\si,\sf)}.
\end{split}
\end{equation}
\item For any $\si \le \sf$ and $\Delta t\geq 0$, we have 
\begin{equation}\label{prop: B}
  B(\si,\sf) =
  \begin{cases}
    B(\si- \Delta t,\sf -\Delta t),
    & \text{if } \si \le \sf < 0, \\
   B(\si + \Delta t ,\sf + \Delta t),
    & \text{if } 0 < \si \le \sf,\\
     B(\si- \Delta t,\sf + \Delta t), 
     & \text{if } \si < 0 \le \sf .  
 \end{cases} 
\end{equation}
\item For any $\si$, $\sf$ and $\Delta t$ satisfying $\si \le \sf \le 0  \le \si+\Delta t \le \sf + \Delta t$, we have 
\begin{equation}\label{prop: B conj} 
  B(\si+\Delta t,\sf+\Delta t) = \overline{  B(\si,\sf) }.
\end{equation}
\end{itemize}
\end{proposition}

\eqref{prop: B and Gs} can be verified by a case-by-case argument under different settings of $\si$ and $\sf$ and its detailed proof is placed in Appendix \ref{app:proof}. \eqref{prop: B} and \eqref{prop: B conj} are the results derived by the definition of the two-point correlation \eqref{def:B}. We remark that due to the existence of $O_s$ in the definition of $G_s^{(0)}(\cdot,\cdot)$, the first equality in \eqref{prop: B and Gs} does not hold when $\si$ or $\sf$ equal to $0$.

\subsection{Integro-differential equation for the propagator}
To derive the integro-differential equation, we begin with calculating the derivative of $G(-t,t)$. By definition \eqref{G -t t}, 
\begin{equation*}
    \begin{split}
& G(-(t+\Delta t),(t+\Delta t))\\
=   & \  \ee^{\ii  (t+\Delta t) H_s} O_s \ee^{- \ii  (t+\Delta t) H_s} +   \sum_{\substack{m=2 \\ m \text{~is even}}}^{+\infty}
  \ii^m \int_{-t \le \sb \le t} \dd \sb 
(-1)^{\#\{\sb < 0\}}  \mathcal{U}^{(0)}(-(t+\Delta t), \sb , t+ \Delta t) \cdot
    \mathcal{L}_b(\sb) \\
 & -  \Bigg( \sum_{\substack{m=2 \\ m \text{~is even}}}^{+\infty}
  \ii^m  \int_{-(t+\Delta t)}^{-t} \dd s_1 \int_{s_1 \le s_2 \cdots \le s_m \le t} \dd s_2 \cdots \dd s_m \times \\
 & \hspace{150pt} \times   (-1)^{\#\{\{s_i\}_{i=2}^m < 0\}}  \mathcal{U}^{(0)}(-(t+\Delta t), \sb , t+ \Delta t) \cdot
    \mathcal{L}_b(\sb) \Bigg) \\
 &   + \Bigg( \sum_{\substack{m=2 \\ m \text{~is even}}}^{+\infty}
  \ii^m  \int^{t+\Delta t}_{t} \dd s_m \int_{-(t+\Delta t)\le s_1\le \cdots s_{m-1} \le s_m} \dd s_1 \cdots \dd s_{m-1} \times \\
 & \hspace{150pt} \times   (-1)^{\#\{\{s_i\}_{i=1}^{m-1} < 0\}}   \mathcal{U}^{(0)}(-(t+\Delta t), \sb , t+ \Delta t) \cdot
    \mathcal{L}_b(\sb) \Bigg).
       \end{split}
\end{equation*}
Here we split all integrals into three parts based on the distribution of the time sequences. Note that a minus sign is added before the second summation above since $s_1$ is restricted within $[-(t+\Delta t),-t]$ in this term and thus $(-1)^{\#\{\sb < 0\}} = -(-1)^{\#\{\{s_i\}_{i=2}^m < 0\}} $. This expression allows us to differentiate $G(-t,t)$ by the definition of the derivative: 
\begin{equation}\label{G t plus delta t}
\begin{split}
 \frac{\dd}{\dd t} G(-t,t) = & \ \lim_{\Delta t \rightarrow 0} \frac{G(-(t+\Delta t),(t+\Delta t))- G(-t,t)  }{\Delta t} \\
=& \ \frac{\dd}{\dd t}  \left( \ee^{\ii   t H_s} O_s \ee^{- \ii  t H_s} \right) + \sum_{\substack{m=2 \\ m \text{~is even}}}^{+\infty}
  \ii^m \int_{-t \le \sb \le t} \dd \sb 
(-1)^{\#\{\sb < 0\}}  \frac{\dd}{\dd t} \mathcal{U}^{(0)}(-t, \sb , t) \cdot \mathcal{L}_b(\sb)\\
 & - \Bigg( \sum_{\substack{m=2 \\ m \text{~is even}}}^{+\infty}
  \ii^m  \int_{-t\le s_2\le \cdots \le s_m \le t} \dd s_2 \cdots \dd s_m\times \\
  & \hspace{50pt}\times (-1)^{\#\{\{s_i\}_{i=2}^m < 0\}} \mathcal{U}^{(0)}(-t, -t,\underline{s_2 ,\cdots ,s_m}, t) \cdot  \mathcal{L}_b(-t,\underline{s_2,\cdots,s_m})\Bigg) \\
 & + \Bigg( \sum_{\substack{m=2 \\ m \text{~is even}}}^{+\infty}
  \ii^m  \int_{-t \le s_1\le \cdots \le s_{m-1} \le t} \dd s_1 \cdots \dd s_{m-1} \times \\
  & \hspace{50pt} \times  (-1)^{\#\{\{s_i\}_{i=1}^{m-1} < 0\}} \mathcal{U}^{(0)}(-t, \underline{s_1,\cdots, s_{m-1}},t,t) \cdot \mathcal{L}_b(\underline{s_1,\cdots,s_{m-1}},t) \Bigg).
\end{split} 
\end{equation}
Using the definition \eqref{def U0} of $\mathcal{U}^{(0)}$, the derivative in the first series is computed by 
\begin{align*}
\frac{\dd}{\dd t} \mathcal{U}^{(0)}(-t, \sb , t) = & \  \frac{\dd}{\dd t} \left(G_s^{(0)}(s_m, t) \right) W_s G_s^{(0)}(s_{m-1}, s_{m}) W_s
  \cdots W_s G_s^{(0)}(s_1, s_2) W_s G_s^{(0)}(-t, s_1)\\
  & + G_s^{(0)}(s_m, t)  W_s G_s^{(0)}(s_{m-1}, s_{m}) W_s
  \cdots W_s G_s^{(0)}(s_1, s_2) W_s \frac{\dd}{\dd t}  \left( G_s^{(0)}(-t, s_1) \right) \\
  = & \ \ii H_s  \mathcal{U}^{(0)}(-t, \sb , t) - \ii \mathcal{U}^{(0)}(-t, \sb , t)H_s .
\end{align*}
Note that $\frac{\dd}{\dd t}  \left( \ee^{\ii   t H_s} O_s \ee^{- \ii  t H_s} \right) = \ii H_s \ee^{\ii   t H_s} O_s \ee^{- \ii  t H_s} - \ii  \ee^{\ii   t H_s} O_s \ee^{- \ii  t H_s} H_s$, which yields
\begin{equation}
\frac{\dd}{\dd t}  \left( \ee^{\ii   t H_s} O_s \ee^{- \ii  t H_s} \right) + \sum_{\substack{m=2 \\ m \text{~is even}}}^{+\infty}
  \ii^m \int_{-t \le \sb \le t} \dd \sb 
(-1)^{\#\{\sb < 0\}}  \frac{\dd}{\dd t} \mathcal{U}^{(0)}(-t, \sb , t) \cdot \mathcal{L}_b(\sb)=\ii [H_s,G(-t,t)].
\end{equation}
As for the other two series on the right-hand side of \eqref{G t plus delta t}, we can simplify them by using
\begin{displaymath}
\mathcal{U}^{(0)}(-t, -t,\sb, t) = W_s  \mathcal{U}^{(0)} (-t,\sb, t), \qquad
\mathcal{U}^{(0)}(-t, \sb,t,t) = \mathcal{U}^{(0)} (-t,\sb, t) W_s.
\end{displaymath}
Summarizing all the simplifications of \eqref{G t plus delta t}, we obtain
\begin{multline}\label{dyson int diff eq}
\frac{\dd }{\dd t}G(-t,t) = \ii [H_s,G(-t,t)] +
 \sum_{\substack{m=1 \\ m \text{~is odd}}}^{+\infty}
  \ii^{m+1}  \Big(   \int_{-t\le \sb \le t} \dd \sb  (-1)^{\#\{\sb < 0\}} W_s \mathcal{U}^{(0)}(-t, \sb , t)  \Ls_b(\sb,t)\\
  -  \int_{-t\le \sb \le t} \dd \sb  (-1)^{\#\{\sb < 0\}}  \mathcal{U}^{(0)}(-t, \sb , t) W_s \Ls_b(-t,\sb) \Big).
\end{multline}
Note that $m$ takes odd values in \eqref{dyson int diff eq} since the underlined time sequences in \eqref{G t plus delta t} have odd numbers of components.
The equation \eqref{dyson int diff eq} has already provided us an integro-differential equation to work on. However, we may make further simplification by combining the two integrals into one using the following lemma:
\begin{lemma}\label{lemma:dqmc}
For any time sequence $-t<s_1<\cdots< s_{m}<t$, define $s'_{i} = -s_{m+1-i}$ for $i=1,\cdots,m$. Then $-t <s'_1<\cdots< s'_m<t$ and
\begin{align}
 \mathcal{U}^{(0)}(-t,\sb',t) = &  \ \mathcal{U}^{(0)}(-t,\sb,t)^\dagger \text{~for $s_i \neq 0$, $i = 1,\cdots,m$}, \label{lemma:U0}\\
\Ls_b(-t,\sb')  = & \ \overline{ \Ls_b(\sb,t) }. \label{lemma:Lb}
\end{align}
\end{lemma}

The statement \eqref{lemma:U0} for the system associated $ \mathcal{U}^{(0)}$ can be checked by 
\begin{displaymath}
\begin{split}
 \mathcal{U}^{(0)}(-t,\sb',t)  = & \  G_s^{(0)}(s'_m, t) W_s G_s^{(0)}(s'_{m-1}, s'_{m}) W_s
  \cdots W_s G_s^{(0)}(s'_1, s'_2) W_s G_s^{(0)}(-t, s'_1)\\
  = & \  G_s^{(0)}(-s_1, t) W_s G_s^{(0)}(-s_{2}, -s_{1}) W_s
  \cdots W_s G_s^{(0)}(-s_m, -s'_{m-1}) W_s G_s^{(0)}(-t, -s_m)\\
  = & \  G_s^{(0)}(-t, s_1)^\dagger W_s G_s^{(0)}(s_{1}, s_{2})^\dagger W_s
  \cdots W_s G_s^{(0)}(s_{m-1}, s_{m})^\dagger W_s G_s^{(0)}(s_m, t)^\dagger \\
  = & \ \mathcal{U}^{(0)}(-t,\sb,t)^\dagger 
  \end{split}
\end{displaymath}
using \eqref{prop: B and Gs}. The equation \eqref{lemma:Lb} can also be verified using \eqref{prop: B and Gs}. The rigorous proof can be found in Appendix \ref{app:proof}. 

Now we apply the change of variables as shown in Lemma \ref{lemma:dqmc} to the second integral in \eqref{dyson int diff eq}. Note that \eqref{lemma:U0} holds almost everywhere in the domain of integration. We then have 
\begin{equation}\label{change of variable}
   \begin{split}
&\int_{-t\le \sb \le t} \dd \sb  (-1)^{\#\{\sb < 0\}}  \mathcal{U}^{(0)}(-t, \sb , t) W_s \Ls_b(-t,\sb) \\
= \ & \int_{-t\le \sb' \le t}  \dd \sb' (-1)^{\#\{\sb' > 0\}}   \mathcal{U}^{(0)}(-t, \sb' , t)^\dagger W_s \overline{ \Ls_b(\sb',t)} \\
= \ & -\int_{-t\le \sb' \le t} \dd \sb' (-1)^{\#\{\sb' < 0\}} \left( W_s \mathcal{U}^{(0)}(-t, \sb' , t)    \Ls_b(\sb',t)  \right)^\dagger.
   \end{split}
\end{equation}
In the last equality above, we have used the fact that $\sb'=(s'_1,\cdots,s'_m)$ has odd number of components and thus $(-1)^{\#\{\sb' > 0\}} = -(-1)^{\#\{\sb' < 0\}}$ for almost every $\sb'$. Inserting \eqref{change of variable} back to \eqref{dyson int diff eq}, we reach a simpler integral-differential equation for $G(-t,t)$:
\begin{proposition}
The propagator $G(-t,t)$ satisfies the integro-differential equation
\begin{equation}\label{dyson int diff eq 2}
\frac{\dd }{\dd t}G(-t,t) = \ii [H_s,G(-t,t)] + \sum_{\substack{m=1 \\ m \text{~is odd}}}^{+\infty}  \ii^{m+1} \int_{-t\le \sb \le t} \dd \sb  (-1)^{\#\{\sb < 0\}} (\mathcal{K}(\sb,t) +\mathcal{K}(\sb,t)^\dagger ) 
\end{equation}
for $t>0$, where 
\begin{displaymath}
\mathcal{K}(\sb,t) =  W_s \mathcal{U}^{(0)}(-t, \sb , t)  \Ls_b(\sb,t).
\end{displaymath}
\end{proposition}
Based on the evolution equation \eqref{dyson int diff eq 2}, one can consider solving $G(-t,t)$ iteratively using Runge-Kutta type methods. To avoid large values of $m$ in the computation, we truncate the series up to a certain odd integer, and evaluate the high-dimensional integrals stochastically via Monte Carlo approximation. Compared to the original bare dQMC for the Dyson series \eqref{bare dqmc}, solving \eqref{dyson int diff eq 2} should be more efficient as $m$ decreases by $1$ for each term of the summation. We also point out that the numerical methods based on the integro-differential equation preserve the Hermitian property \eqref{G hermitian} of $G(-t,t)$ as one can easily check that the right-hand side of \eqref{dyson int diff eq 2} is always Hermitian under Monte Carlo approximation, while this is not guaranteed by bare dQMC \eqref{bare dqmc} and may be badly violated when the number of samples $\Ms$ is insufficient. Moreover, since the equation provides us the time evolution of $G(-t,t)$, we are now able to reuse the calculated bath influence functionals which will give the major improvement on the efficiency of the algorithm. Our numerical method will be detailed in the following section.

\subsection{Numerical method}
To discretize \eqref{dyson int diff eq 2}, we consider a numerical scheme inspired by the second-order Heun's method. For a general ordinary differential equation
\begin{displaymath}
\frac{\mathrm{d}}{\mathrm{d}t} u(t) = f(t, u(t)), \quad t \in [0,t_{\max}],
\end{displaymath}
the scheme reads
\begin{equation}\label{Heun}
\begin{split}
U_i^* &= U_{i-1} + h f(t_{i-1}, U_{i-1}), \\
U_i &= \frac{1}{2} (U_{i-1} + U_i^*) + \frac{1}{2} h f(t_i, U_i^*),
\end{split}
\end{equation}
where $h$ is the time step length, $t_i = i\cdot h$, and $U_i$ is
the numerical approximation of $u(t_i)$. For our integro-differential equation, the sums over high-dimensional integrals should be evaluated in the same way as the bare dQMC \eqref{bare dqmc} using Monte Carlo approximation. In the $i$th time step, suppose we have $\Ms_i$ samples of time sequences $\Ss_i = \{\sb^{(j)}_i\}_{j=1}^{\Ms_i}$ drawn from the domain
\begin{equation}
T_i = \bigcup_{\substack{m=1 \\ m \text{~is odd}}}^{\bar{M}} T^{(m)}_i,
\end{equation}
where $T^{(m)}_i$ is the $m$-dimensional simplex defined by
\begin{equation}
T^{(m)}_i:=\{ \sb=(s_1,\cdots,s_m) \ | \ {-t_i} \le s_1 \le \cdots \le s_m \le t_i \},
\end{equation}
and each sampled time sequence satisfies the probability density function $\P_i(m,\sb)$ for $m=1,3,\cdots,\bar{M}$. Thereby, the scheme coupling Heun's method with Monte Carlo integration to approximate $G(-t_i,t_i)$ is formulated as 
\begin{equation}\label{dyson int diff eq scheme}
  \begin{split}
& G_i^* =    G_{i-1} + h\ii [H_s,G_{i-1}]+ \frac{h}{{\Ms}_{i-1}}\sum^{{\Ms}_{i-1}}_{j=1}  \frac{1}{ \P_{i-1}( m_{i-1}^{(j)},\sb_{i-1}^{(j)})  }\cdot \ii^{m_{i-1}^{(j)}+1 }  \times \\
& \hspace{50pt} \times  (-1)^{\#\{\sb_{i-1}^{(j)} < 0\}} (\mathcal{K}(\sb_{i-1}^{(j)},t_{i-1}) +\mathcal{K}(\sb_{i-1}^{(j)},t_{i-1})^\dagger ), \\  
& G_i =   \frac{1}{2}(G_{i-1} + G_i^*) + \frac{1}{2}h\ii[H_s,G_i^*] + \frac{h}{{2\Ms}_{i}}\sum^{{\Ms}_{i}}_{j=1} \frac{1}{ \P_{i}( m_{i}^{(j)},\sb_i^{(j)} ) } \cdot  \ii^{m_{i}^{(j)}+1 } \times  \\
&  \hspace{50pt}  \times  (-1)^{\#\{\sb_{i}^{(j)} < 0\}} (\mathcal{K}(\sb_{i}^{(j)},t_{i}) +\mathcal{K}(\sb_{i}^{(j)},t_{i})^\dagger ), \text{~for~} \sb^{(j)}_i \in \Ss_i
  \end{split}
\end{equation}
for $i = 1,2,\cdots,N$ where $Nh = t_{\max}$ with initial condition $G_{0} = O_s$. The set of samples $\Ss_i$ are drawn independently according to the distribution $\P_i$. We remark that one can apply higher order schemes to achieve better order of accuracy with respect to step length $h$. Throughout the current work, however, we use the Heun's method which can already provide satisfactory numerical results. The accuracy of discretization will be verified by numerical tests later in Section \ref{sec:accuracy}. We also refer readers to the numerical experiments in \cite[Section 7]{Cai2020}, where Heun's method is applied to a number of spin-boson simulations and shows good performance.

The major computational cost lies in the evaluation of $\mathcal{K}(\sb,t_i)$ in each time step. While evaluating each $\mathcal{K}(\sb,t_i)$, 
the bath influence functional $\Ls_b(\sb,t_i)$ is generally much more expensive than the $\mc{U}^{(0)}(-t_i,\sb,t_i)$, especially when $m$ is large. In fact, the computational cost of $\Ls_b$, which is essentially the hafnian of a matrix \cite{Barvinok99}, grows at least exponentially with respect to $m$ using some recent indirect methods such as Bj\"orklund's algorithm \cite{Bjorklund2019} or the inclusion-exclusion principle \cite{Yang2021}, while the cost of $\mc{U}^{(0)}$ grows only linearly with $m$ since $\mc{U}^{(0)}$ is a product of $2m+1$ matrices as defined in \eqref{def U0}. A comparison of the computational time for these two parts will be performed later in Section \ref{sec:time complexity}.

Due to the high computational cost of $\Ls_b$, the purpose of this paper is to reduce the number of bath influence functionals to be computed during the evolution of $G(-t,t)$. While the straightforward application of the numerical scheme \eqref{dyson int diff eq scheme} requires computation of different bath influence functionals in different time steps, by the invariance of the two-point bath correlation given in Propositions \eqref{prop: B and Gs} to \eqref{prop: B conj}, we can actually reuse some bath influence functionals that have been calculated in previous time steps to improve the overall efficiency. This idea utilizes the following property of $\Ls_b$, which can be easily derived from \eqref{prop: B}:   
\begin{proposition}\label{prop: L}
Given $\sb = (s_1,\cdots, s_m)\in T^{(m)}_i$ for $i=1,\cdots,N-1$ and odd number $m=1,3,\cdots,\bar{M}$, define the operator $\Is_j(\sb) = (\st_1,\cdots, \st_m)$ such that
\begin{equation} \label{s tilde}
  \st_k =
  \begin{cases}
   s_k + jh,
    & \text{if } s_k \geq 0, \\
   s_k-jh,
    & \text{if } s_k < 0
 \end{cases}
\end{equation}
for $k=1,\cdots,m$ and $j = 0,1,\cdots,N-i$. We have $\Is_j(\sb) \in T^{(m)}_{i+j}$ and
\begin{equation*}
\Ls_b(\Is_j(\sb),t_{i+j}) = \Ls_b(\sb,t_i).
\end{equation*}
\end{proposition}

This proposition shows that a class of bath influence functionals has the same value, and thus we just need to compute one of them if multiple influence functionals appear in our computation.
    To illustrate how such reuse can be applied to the scheme \eqref{dyson int diff eq scheme}, we consider the following simple example, where we only sample one time sequence with $m=1$ (so the sequence actually reduces to a point) in each time step and consider the time evolution of the scheme up to $t=3h$:
    \begin{itemize}
        \item[(\textbf{i})] in the first time step, we pick a sample $s_1 \in (-h,h)$. Here we assume $s_1$ is negative which can be denoted by the black dot in the top panel of Figure \ref{fig:reuse}. The corresponding bath influence functional $\Ls_b(s_1,h) = B(s_1,h)$ is then calculated and can be denoted by blue arc;
        \item[(\textbf{ii})] in the second time step, $\Is_1(s_1) = s_1 - h$ is a sample in $T^{(1)}_2$ whose bath influence functional can be directly obtained from $\Ls_b(\Is_1(s_1),2h) = \Ls_b(s_1,h)$ according to Proposition \ref{prop: L}. Such reuse of computed bath influence functionals can be visualized as a stretch of the blue arc by length $h$ in Figure \ref{fig:reuse}, and the value of the blue arc is invariant after being stretched. In addition to reuse of calculations, we sample a new time point $s_2 \in (-2h,2h)$ and calculate $\Ls_b(s_2,2h)$. In Figure \ref{fig:reuse}, we assume $s_2$ is positive and $\Ls_b(s_2,2h)$ is represented by the red arc;
        \item[(\textbf{iii})] at $t=3h$, the blue arc can be further stretched by another time step $h$ and the value remains the same, meaning that we again obtain the bath influence functional directly using $\Ls_b(\Is_2(s_1),3h) = \Ls_b(s_1,h)$ where $ \Is_2(s_1) = s_1 -2h \in T^{(1)}_3$. Similarly, we can also reuse $\Ls_b(\Is_1(s_2),3h) = \Ls_b(s_2,2h)$ with $\Is_1(s_2) = s_2 + h \in T^{(1)}_3$, which corresponds to shifting the red arc to the right by $h$. Afterwards, we draw another new sample $s_3 \in (-3h,3h)$ and calculate $\Ls_b(s_3,3h)$ denoted by the green arc.
    \end{itemize}
For general $m$, this reuse of bath influence functionals can be similarly understood by replacing the arcs by the summation of diagrams such as in \eqref{eq:all pairs example}. We remark that such invariance does not hold form the system functional $\mc{U}^{(0)}$, which does not have a similar property as Proposition \ref{prop: L} due to the existence of $O_s$ in its definition.  

\input{images/property_L_revised}

As can be observed from Figure \ref{fig:reuse}, given any time sequence $\sb_j$ at the $j$th time step for $j< i$, shifting or stretching it to $\Is_{i-j}(\sb_j)$ always moves the nodes away from $t=0$ by at least length $h$. This means all the samples obtained by stretching or shifting have no time points falling between $-h$ and $h$. As a result, the samples for the $i$th time step cannot be only inherited from previous time steps. To complete the sampling of $T_i$, we also need to draw extra samples from $\sT_i =\bigcup_{\substack{m=1 \\ m \text{~is odd}}}^{\bar{M}} \sT^{(m)}_i$ where     
\begin{equation}\label{T hat}
 \sT^{(m)}_{i}  = \left\{ ( s_1,\cdots, s_m) \in T^{(m)}_i \ \big| \ \exists \ s_j \text{~such that~} {-h} < s_j < h  \right\}.
\end{equation} 
For example, in Figure \ref{fig:reuse}, the nodes building up the red diagrams in $(-2h,2h)$ and green diagrams $(-3h,3h)$ should be newly drawn from $\sT_{2}$ and $\sT_{3}$ respectively since these diagrams can never be obtained from shifting or stretching diagrams at previous time steps. Based on the definition \eqref{T hat}, we may express $T^{(m)}_i$ as 
\begin{displaymath}
 T^{(m)}_i = \bigcup_{j=1}^{i} \Is_{i-j}(\sT^{(m)}_j)
\end{displaymath}
where $\Is_{i-j}(\sT^{(m)}_j)$ is the collection of time sequences which are shifted or stretched from $j$th step: 
 \begin{displaymath}
\Is_{i-j}(\sT^{(m)}_j) = \{  \Is_{i-j}(\sb) \ | \ \sb \in \sT^{(m)}_j  \}.
\end{displaymath}
One may easily see that $\Is_{i-j}(\sT^{(m)}_j)$ are pairwise disjoint for $j=1,\cdots,i$ and thus 
\begin{equation}\label{T vol relation}
\sum_{j=1}^i |\sT^{(m)}_j| = \sum_{j=1}^i |\Is_{i-j}(\sT^{(m)}_j)| = |T^{(m)}_i|.
\end{equation}
Hence the volume of each $\sT^{(m)}_j$ can be calculated by 
\begin{equation}\label{T star vol}
|\sT^{(m)}_j| = |T^{(m)}_{j}| - |T^{(m)}_{j-1}| =  \frac{1}{m!} [(2t_j)^m - (2t_{j-1})^m].
\end{equation}

To implement the numerical scheme \eqref{dyson int diff eq scheme}, we sample time sequences $\hat{\Ss}_i \subset \sT_i$ in each step and evaluate the corresponding bath influence functionals. Afterwards, we construct $\Ss_i$ by combining the new samples $\hat{\Ss}_i$ with the old samples $\Is_{i-j}(\hat{\Ss}_j)$ for $j=1,\cdots,i-1$ whose bath influence functionals can be directly reused by Proposition \ref{prop: L}, and then evaluate $G_i$ according to \eqref{dyson int diff eq scheme}. Such a procedure is described by the Algorithm \ref{algo:dyson}.   
\begin{algorithm}[ht]
  \caption{Dyson series}\label{algo:dyson}
  \begin{algorithmic}[1]
\State \textbf{input}  $\hat{\Ss}_{i}= \{\sb^{(j)}_i\}_{j=1}^{\hat{\Ms}_i} \subset \sT_{i}$ for $i=1,\cdots,N$
  \medskip
    \State Set $G_{0}  \gets \text{Id}$ 
      \medskip
  \For{$i$ from $1$ to $N$}  
  \medskip
  \State Compute $\hat{L}_i =  \{\Ls_b(\sb^{(j)}_i,t_i)\}_{j=1}^{\hat{\Ms}_i}$
    \medskip
\State Set $\Ss_i \gets \bigcup_{j=1}^i \Is_{i-j}(\hat{\Ss}_j)$ \Comment{\emph{Shift/stretch samples}}
  \medskip
  \State Set $L_i \gets \bigcup_{j=1}^i \hat{L}_j$   \Comment{\emph{Reuse bath calculation}}
  \medskip
\State Compute $G_i$ by scheme \eqref{dyson int diff eq scheme} based on $\Ss_i $ and $L_i$ 
 \EndFor
  \medskip
  \State \textbf{return} $G_{i}$ for $i = 1,\cdots, N$
   \end{algorithmic}
\end{algorithm}

To complete the implementation, we need to specify the sampling strategy for the input $\hat{\Ss}_i$, which is associated with the probability density function $\P_i(m,\sb)$ in \eqref{dyson int diff eq scheme}. Ideally, the number of samples in $\hat{\Ss}_{i}$ should be proportional to the integral of the absolute value of the bath influence functional: 
\begin{equation}
\hat{\Ms}_i^{(m)} \propto   \int_{ \sb\in \sT^{(m)}_i } \dd \sb  |\Ls_b(\sb,t_i)| =  \int_{\sb\in \sT^{(m)}_i } \dd \sb \left| \sum_{\mf{q} \in \mQ(\sb,t_i)} \prod_{(s_j,s_k) \in \mf{q}} B(s_j,s_k) \right|.
\end{equation}
In practice, as the integral is difficult to evaluate, we replace $B(s_j, s_k)$ by an empirical constant $\mathcal{B} \in (0, \max |B|)$, so that
\begin{equation}\label{P s* in T*}
\hat{\Ms}_i^{(m)} =     \frac{\hat{\Ms}^{(1)}_1}{\hat{\lambda}} \cdot  |\sT^{(m)}_i| \cdot m!! \mathcal{B}^{\frac{m+1}{2}} 
=    \frac{\hat{\Ms}^{(1)}_1}{\hat{\lambda}} \cdot  \frac{(2t_i)^m - (2t_{i-1})^m}{(m-1)!!} \cdot \mathcal{B}^{\frac{m+1}{2}}
\end{equation}
where $\hat{\lambda} = 2 \mathcal{B} h$ is the normalizing factor. In the numerical implementation, one may first assign $\hat{\Ms}^{(1)}_1 = \hat{\Ms}_0$, and the other $\hat{\Ms}^{(m)}_i$ can then be set as the nearest integer to the right-hand side of the formula above. Afterwards, we generate each time sequence $\sb = (s_1, \cdots, s_m) \in \sT^{(m)}_i$ by drawing a sample from the uniform distribution $U(\sT^{(m)}_i)$. The following theorem provides the explicit expression for the probability density $\P_i(m,\sb)$ appearing in scheme \eqref{dyson int diff eq scheme}:

\begin{proposition}\label{thm:P}
For any $i =1,2,\cdots,N$ and $m = 1,3,\cdots,\bar{M}$, $\P_i(m,\sb)$ is given by  
\begin{equation}\label{P}
\P_i(m,\sb) =  \frac{1}{\lambda_i} \cdot m!! \mathcal{B}^{\frac{m+1}{2}}
\end{equation}
where 
\begin{displaymath}
\lambda_i =  \sum_{\substack{m'=1 \\ m' \text{~is odd}}}^{\bar{M}} \frac{(2t_i)^{m'}}{(m'-1)!!}\cdot \mathcal{B}^{\frac{m'+1}{2}}  
\end{displaymath}
\end{proposition}
\begin{proof}
For any time sequence $\sb$ which is obtain by either reuse or newly sampling in each step, we have $\P(\sb \in T^{(m)}_i) \propto \Ms^{(m)}_{i}$ where $\Ms^{(m)}_{i}$ is the number of time sequences with $m$ components in $i$th step. According to our sampling strategy,
\begin{align*}
\Ms^{(m)}_{i} =  \sum_{j = 1}^{i} {\hat{\Ms}}^{(m)}_{j}  = \frac{\hat{\Ms}_0}{\hat{\lambda}} \cdot  m!! \mathcal{B}^{\frac{m+1}{2}} \cdot  \sum_{j=1}^i |\sT^{(m)}_j| =   \frac{\hat{\Ms}_0}{\hat{\lambda}} \cdot m!! \mathcal{B}^{\frac{m+1}{2}} \cdot |T^{(m)}_i| 
\end{align*}
where we have used the relation \eqref{T vol relation} for the last inequality. Note that the time sequences $\Ss_i$ used in scheme \eqref{dyson int diff eq scheme} are constructed the samples drawn from the pairwise disjoint $U[\Is_{i-j}(\sT^{(m)}_{j})]$, and the number of these samples locating in each $\Is_{i-j}(\sT^{(m)}_{j})$ is proportional to the volume $ |\Is_{i-j}(\sT^{(m)}_{j})|$ according to \eqref{P s* in T*}. Therefore, any time sequence in $\Ss_i$ can be considered as a sample drawn in $U[T^{(m)}_i]$ and thus we reach the conclusion \eqref{P} by 
\begin{align*}
\P_i(m,\sb) = & \ \P(\sb \in T^{(m)}_i | \sb \in T_i)  \cdot \frac{1}{|T^{(m)}_i|}\\
 = & \ \frac{\P(\sb \in T^{(m)}_i)}{\sum_{\substack{m'=1 \\ m' \text{~is odd}}}^{\bar{M}}\P(\sb \in T^{(m')}_i)}\cdot \frac{1}{|T^{(m)}_i|} = \frac{1}{\sum_{\substack{m'=1 \\ m' \text{~is odd}}}^{\bar{M}} |T^{(m')}_i| \cdot m'!! \mathcal{B}^{\frac{m'+1}{2}}} \cdot m!! \mathcal{B}^{\frac{m+1}{2}}.\qedhere
\end{align*}
\end{proof}

\subsection{Implementation of Algorithm \ref{algo:dyson} with low memory cost}
In general, the reuse of bath calculations described in Algorithm \ref{algo:dyson} requires storing of all time sequences (Line 5) as well as bath influence functionals (Line 6) in a simulation, which will lead to a high memory cost when the number of samples is large. However for Dyson series, the linearity of its governing equation \eqref{dyson int diff eq 2} allows us to implement the reuse algorithm at a much lower memory cost. To begin with, we apply the scheme \eqref{dyson int diff eq scheme} recursively and get the following explicit formula for any $G_i$:
\begin{equation}\label{Gi formula}
 G_i = \tilde{\alpha}^i O_s + \frac{1}{2}h  \sum_{k = 1}^{i-1} \tilde{\alpha}^{i-k-1} (\alpha + \tilde{\alpha})(\beta_k + \beta_k^\dagger) + \frac{1}{2}h(\beta_i + \beta_i^\dagger)
\end{equation}
where the operator $\alpha = 1 + \ii h [H_s,\cdot]$ and $\tilde{\alpha} = \frac{1}{2}(1+\alpha^2)$. $\beta_i$ is the average of Monte Carlo samples $\sb = (s_1,\cdots,s_m)$ 
\begin{displaymath}
 \beta_i = \frac{1}{{\Ms}_{i}}\sum^{{\Ms}_{i}}_{j=1}   \gamma_i(\sb_{i}^{(j)})\cdot  \mathcal{U}^{(0)}(-t_i,\sb_{i}^{(j)},t_i )\Ls_b(\sb_{i}^{(j)},t_i) \quad \text{with}  \quad \gamma_i(\sb) = \frac{1}{ \P_{i}( m_,\sb ) } \cdot  \ii^{m+1 } \cdot  (-1)^{\#\{\sb < 0\}}.
\end{displaymath}

Note that the direct evaluation of $G_i$ by \eqref{Gi formula} requires the storage of all reusable $\Ls_b$. To avoid this, we consider the following resummation of $\beta_i$ according to where the samples are originally generated:
\begin{equation}\label{beta}
 \beta_i =  \theta_{i1} + \theta_{i2} + \cdots \theta_{ii} 
\end{equation}
where the partial sum
\begin{displaymath}
   \theta_{ik} := \frac{1}{{\Ms}_{i}}\sum^{{\hat{\Ms}}_{k}}_{j=1}  \gamma_i\left(\Is_{i-k}(\hat{\sb}_{k}^{(j)})\right)\cdot  \mathcal{U}^{(0)}(-t_i,\Is_{i-k}(\hat{\sb}_{k}^{(j)}),t_i )\Ls_b(\Is_{i-k}(\hat{\sb}_{k}^{(j)}),t_i)
\end{displaymath}
stands for the part of calculations where the samples are shifted or stretched from $k$th step. Note that any $k$th column of 
\begin{equation}\label{partial sum}
    \begin{split}
    &\theta_{11},\\
    &\theta_{21},\theta_{22},\\
    & \cdots, \cdots, \cdots,\\
    & \theta_{i1},\theta_{i2},\cdots,\cdots,\theta_{ii},\\
    & \cdots,\cdots, \cdots ,\cdots,\cdots,\\
    &  \theta_{N1},\theta_{N2},\cdots,\cdots,\cdots,\theta_{NN}.
    \end{split}
\end{equation}
share the same bath influence functionals with the value $\{\Ls_b(\hat{\sb}_{k}^{(j)},t_k)\}_{j=1}^{\hat{\Ms}_k}$ according to the Proposition \ref{prop: L}. Therefore, once we have computed one $\Ls_b(\hat{\sb}_{k}^{(j)},t_k)$, it is added to all $\theta_{ik}$ for $i = k,\cdots,N$ and then can be discarded and thus we need to only store one single bath influence functional to obtain all $\theta_{kk'}$ for $1 \le k' \le k \le N$. In the end, the total memory cost for computing $G_i$ for $i=1,\cdots,N$ will only be the storage of these $\theta_{kk'}$, which are essentially $(N+1)N/2$ two-by-two matrices.

\subsection{Analysis on computational cost}
   \label{sec:comp cost dqmc}
To conclude the discussion on the summation of Dyson series, we examine the computational cost that is saved by reusing the bath influence functionals. Specifically, we consider the ratio $1- \#\{\ssb^{(m)}\}/\#\{\sb^{(m)}\}$ for various $m$ where  
\begin{align*}
    \#\{\ssb^{(m)}\} = & \  \hat{\mathcal{M}}_1^{(m)} + \hat{\mathcal{M}}_2^{(m)} + \cdots + \hat{\mathcal{M}}_N^{(m)}, \\ 
     \#\{\sb^{(m)}\} = & \ \mathcal{M}_1^{(m)} + \mathcal{M}_2^{(m)} + \cdots + \mathcal{M}_N^{(m)}.
\end{align*}
Here $\#\{\ssb^{(m)}\}$ denotes the number of $(m+1)$-point bath influence functionals that one needs to evaluate up to $N$th time step in our algorithm, and $\#\{\sb^{(m)}\}$ denotes the corresponding number if all bath influence functionals are to be calculated. For example in Figure \ref{fig:reuse}, we have $\#\{\ssb^{(1)}\} = 3$ as the blue and red diagrams need to be computed only once. However, without reusing the existing information, one then has to draw all the time sequences independently and compute the corresponding $\#\{\sb^{(1)}\} = 6$ (3 blue arcs, 2 red arcs and 1 green arc) bath influence functionals. Therefore, we have achieved a $50\%$ reduction of the computational cost in this example. In general, by \eqref{P s* in T*} we have        
\begin{displaymath}
 \begin{split}
\#\{\ssb^{(m)}\} = & \  \frac{\hat{\Ms}_0}{\hat{\lambda}} \cdot \left(|\sT^{(m)}_1| + | \sT^{(m)}_2| + \cdots + | \sT^{(m)}_N| \right)  \cdot m!! \mathcal{B}^{\frac{m+1}{2}} \\
= & \ \frac{\hat{\Ms}_0}{\hat{\lambda}} \cdot |T^{(m)}_N| \cdot m!! \mathcal{B}^{\frac{m+1}{2}} = \frac{\hat{\Ms}_0}{\hat{\lambda}} \cdot \frac{\sqrt{\mathcal{B}}(2\sqrt{\mathcal{B}} t_N)^{m}}{(m-1)!!}
 \end{split}
\end{displaymath}
and 
\begin{equation}\label{num s dyson}
 \begin{split}
 \#\{\sb^{(m)}\} =  & \ \frac{\hat{\Ms}_0}{\hat{\lambda}} \cdot\left(N| \sT^{(m)}_1| + (N-1)|  \sT^{(m)}_2| + \cdots + |  \sT^{(m)}_N | \right)\cdot m!! \mathcal{B}^{\frac{m+1}{2}} \\
= & \ \frac{\hat{\Ms}_0}{\hat{\lambda}} \cdot\sum_{i=1}^N |T^{(m)}_i| \cdot m!! \mathcal{B}^{\frac{m+1}{2}} = \frac{\hat{\Ms}_0}{\hat{\lambda}} \cdot\sum_{i=1}^N  \frac{\sqrt{\mathcal{B}}(2\sqrt{\mathcal{B}} t_i)^{m}}{(m-1)!!}.
 \end{split}
\end{equation}
Hence, for a given $m$, the percentage of the computational cost that one can save is given by  
\begin{equation}\label{save percent dqmc}
1 - \frac{  \#\{\ssb^{(m)}\}}{\#\{\sb^{(m)}\} } =  1 - \frac{N^m}{1^m + 2^m + \cdots + N^m}=: R^{(m)}(N) 
\end{equation}
which only relies on the number of time steps $N$. 

Below we plot the graphs of $R^{(m)}$ (dashed lines) for various $m$ up to $t=5$ with the time step length $h = 0.05$ in Figure \ref{fig:percent dqmc}, which are all monotonically increasing and thus one may benefit a higher reduction of computational cost from the bath calculation reuse for longer time simulations. In addition, the curves become lower as $m$ grows, indicating that the bath influence functionals with smaller $m$ are reused more frequently for fixed $N$. This observation can be diagrammatically understood in Figure \ref{fig:reuse}. In general, a time sequence $\sb = (s_1,\cdots,s_m)$ with larger $m$ is more likely to have one of its components falling in $(-h,h)$, so that its bath influence functional has to be newly evaluated. However, the value of $m$ usually does not go too large for the purpose of computing Dyson series. When $m=\bar{M} = 25$, one can still expect an around $77\%$ reduction in bath computations at $t=5$. As time further evolves, we may apply the Faulhaber's formula \cite{knuth1993} 
\begin{equation}\label{faulhaber}
    1^m + 2^m + \cdots + N^m \sim \frac{N^{m+1}}{m+1} \text{~as~} N \rightarrow + \infty
    \end{equation}
to get the asymptotic behavior of $R^{(m)}$:
\begin{equation}\label{dyson R asymptotic}
 R^{(m)}(N)  \sim  1 - \frac{m+1}{N}.
\end{equation}

Below we will take into account all choices of $m$ and estimate the overall reduction of the computational cost. Let $\mathcal{T}^{(m)}$ denote the average wall clock time for the evaluation of $\Ls_b(s_1,\cdots,s_m,t)$, the overall savings of the computational time spent on bath computations is then estimated as 
\begin{equation}
    \label{RT}
     R_{\mathrm{T}}(N,h,\mathcal{B}) =  1 - \frac{\sum^{\bar{M}}_{\substack{m = 1 \\ m \text{~is odd}}}  \#\{\ssb^{(m)}\} \cdot \mathcal{T}^{(m)} }{\sum^{\bar{M}}_{\substack{m = 1 \\ m \text{~is odd}}}  \#\{\sb^{(m)}\} \cdot \mathcal{T}^{(m)}}.
\end{equation}
In our implementation, the bath influence functional is computed using a recently proposed fast algorithm based on the inclusion-exclusion principle  \cite[Section 2]{Yang2021}, whose computational complexity is $O(2^m)$. Thereby, asymptotically we have  
\begin{displaymath}
  R_{\mathrm{T}} \sim R^{\text{asy}}_{\mathrm{T}}:= 1 - \frac{\sum^{\bar{M}}_{\substack{m = 1 \\ m \text{~is odd}}}  \#\{\ssb^{(m)}\} \cdot 2^{m}}{\sum^{\bar{M}}_{\substack{m = 1 \\ m \text{~is odd}}}  \#\{\sb^{(m)}\} \cdot 2^{m}} =   1- \frac{ \sum_{m=1}^{\frac{\bar{M}+1}{2}} \frac{(4\sqrt{\mathcal{B}} t_N)^{2m-1}}{(2m-2)!!} }{\sum_{m=1}^{\frac{\bar{M}+1}{2}} \sum_{i=1}^N \frac{(4\sqrt{\mathcal{B}} t_i)^{2m-1}}{(2m-2)!!}}  
\end{displaymath}
where $\mathcal{B}$ is the parameter describing the amplitude of two-point correlation. For large $N$, this can be approximated by
\begin{displaymath}
 R^{\text{asy}}_{\mathrm{T}} \sim  1 - \frac{\bar{M}+1}{N},
\end{displaymath}
which agrees with Figure \ref{fig:percent dqmc} where $R^{\text{asy}}_{\mathrm{T}}$ (solid lines) converges to $R^{(\bar{M})}$ as $t$ grows. This behavior is due to the fact that more bath influence functionals with $m = \bar{M}$ are sampled when $t$ gets larger, and the cost for the evaluation of these $(\bar{M}+1)$-point functionals becomes dominant. For the same reason, the graph of $R^{\text{asy}}_{\mathrm{T}}$ becomes closer to $R^{(\bar{M})}$ as $\mathcal{B}$ increases.

\begin{figure}[ht]
  \includegraphics[scale=0.5]{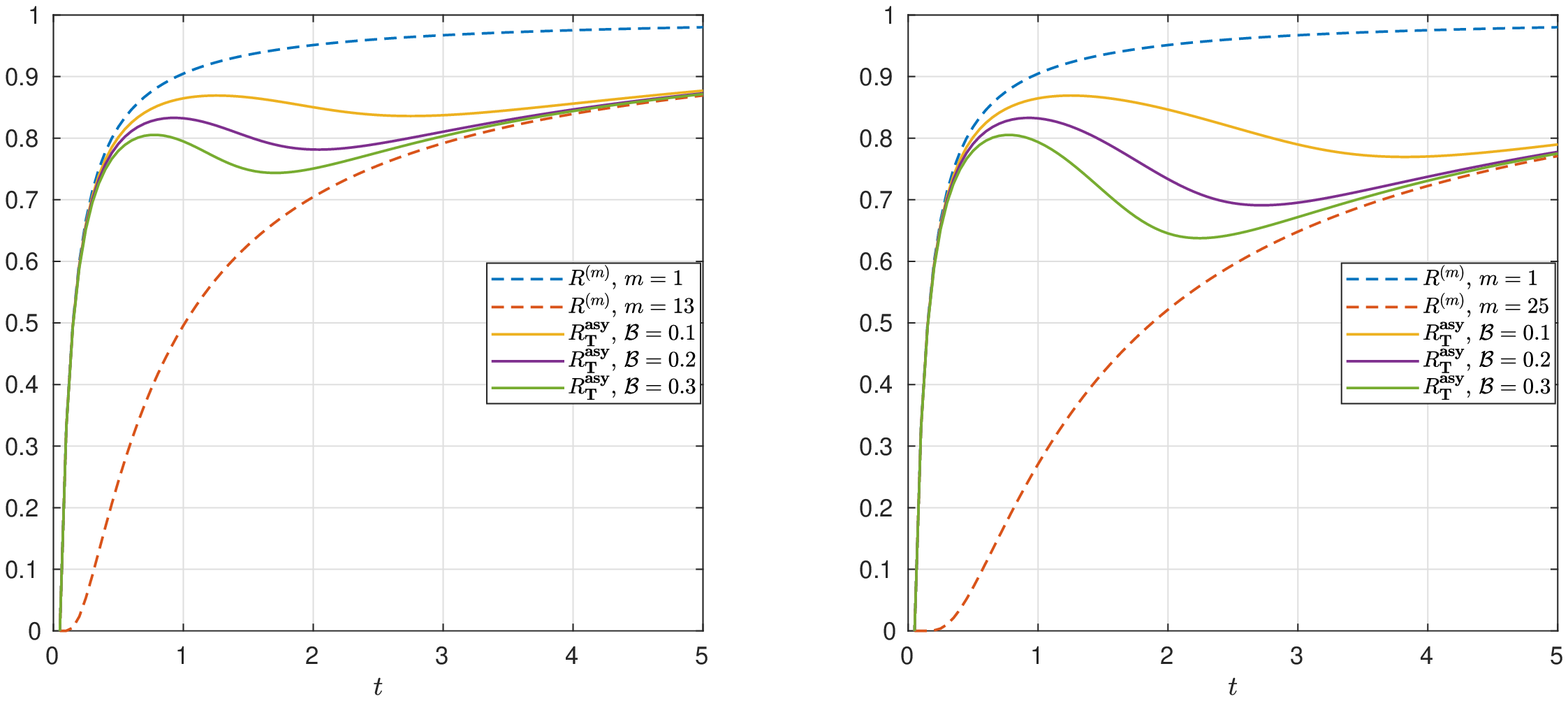}
  \caption{Graphs of $R^{(m)}$ and $R^{\text{asy}}_{\mathrm{T}}$ for Dyson series (left: $\bar{M}=13$, right: $\bar{M} = 25$).}
   \label{fig:percent dqmc}
\end{figure}

\section{Fast implementation of inchworm Monte Carlo method}
   \label{sec:inchworm reuse}
The idea of the fast algorithm for summing Dyson series can also be applied to the inchworm Monte Carlo method introduced in \cite{Cai2020}, which computes the two-variable \emph{full propagator} $G(\si, \sf)$, which generalizes $G(-t,t)$ defined in \eqref{G -t t} to any initial time point $\si \in[-t,t]\backslash \{0\}$ and final time point $\sf \in [\si,t]\backslash \{0\}$. Similar to \eqref{dyson int diff eq 2}, the inchworm method can also be formulated as an integro-differential equation with bath influence functionals of any time series between $\si$ and $\sf$ inside the integral. This structure again allows us to reuse the bath influence functionals computed in previous time steps. Below we will review the formulas of the inchworm Monte Carlo method before introducing our numerical method.

\subsection{Introduction to inchworm Monte Carlo method}

\subsubsection{Full propagator}
The full propagator $G(\si, \sf)$ is formulated by 
\begin{equation} \label{eq:G} 
G(\si, \sf) = G_s^{(0)}(\si,\sf) + \displaystyle \sum_{\substack{m=2\\ m \text{~is even}} }^{+\infty}
  \ii^m \int_{\si \le \sb \le \sf} \dd \sb  (-1)^{\#\{\sb < 0\}}  \mathcal{U}^{(0)}(\si,\sb, \sf) \cdot
    \mathcal{L}_b(\sb) ,
\end{equation}
where $G_s^{(0)}(\si,\sf)$ is given in \eqref{def:Gs}. When $\si = \sf$, it is defined as $G(\si, \sf) = \text{Id}$.
Note that this definition is consistent with the Dyson series \eqref{G -t t} if we set $\si = -t$ and $\sf = t$.
 The following properties of $G(\cdot,\cdot)$ will be found useful later in the numerical method: 
 \begin{proposition}
 \
  \begin{itemize}
      \item \textbf{Shift invariance:} For any $\Delta t \geq 0$, if $\sf + \Delta t < 0$ or $\si  \geq 0$, we have  \begin{equation}\label{shift invariance}
 G(\si+\Delta t,\sf+ \Delta t) = G(\si,\sf).
\end{equation}
\item \textbf{Conjugate symmetry:} For any $-t \le \si \le \sf < t$, we have 
\begin{equation}\label{conjugate symmetry}
    G(-\sf,-\si) = G(\si,\sf)^\dagger
\end{equation}
\item \textbf{Jump condition:} $G(\cdot,\cdot)$ is discontinuous on the line segments
$[-t,0] \times \{0\}$ and $\{0\} \times [-t,0]$ and 
\begin{equation}  \label{jump condition}
    \begin{split}
&\lim_{\sf \rightarrow 0^+} G(\si,\sf) = O_s \lim_{\sf \rightarrow 0^-}G(\si,\sf);  \\
&\lim_{\si \rightarrow 0^-} G(\si,\sf) =  \lim_{\si \rightarrow 0^+}G(\si,\sf)O_s.
  \end{split}
  \end{equation}
  \end{itemize}
 \end{proposition}
 
The rigorous proofs for the statements above are omitted as \eqref{shift invariance} and \eqref{jump condition} can be immediately derived by \eqref{prop: B} and the definition of $\Us^{(0)}$ respectively, and the proof of \eqref{conjugate symmetry} is identical to that of \eqref{lemma:Lb} in Lemma \ref{lemma:dqmc} by changing the variables $\sb$ in the integral to $\sb'$.

\subsubsection{Integro-differential equation for $G(\si,\sf)$}
\label{sec:inchworm int diff eq}
The full propagator has been proved to satisfy the following integro-differential equation \cite{Cai2020}:
\begin{equation}\label{eq: inchworm equation}
    \frac{\partial G(\si,\sf)}{\partial \sf} =   \sgn(\sf)  \left[ \ii H_s G(\si,\sf) +  \sum^{+ \infty}_{\substack{m=1\\ m \text{~is odd~}}} \ii^{m+1} \int_{\si \le \sb \le \sf} \dd \sb   (-1)^{\#\{\sb < 0\}}  W_s \mc{U}(\si,\sb,\sf) \Ls_b^c(\sb,\sf) \right] .
\end{equation}
Here we recall that $W_s$ is the perturbation associated with the system, and $\sgn(\cdot)$ is the sign function. $\mc{U}$ is defined similarly to $\mc{U}^{(0)}$ with the bare propagator $G^{(0)}_s(\cdot, \cdot)$ replaced by the full propagator $G(\cdot,\cdot)$:
\begin{equation}\label{U}
\mc{U}(\si, \sb, \sf) = G(s_{m}, \sf) W_s
  G(s_{m-1}, s_{m}) W_s \cdots W_s G(s_1, s_2) W_s G(\si, s_1).
\end{equation}
The definition of $\Ls_b^c$ is similar to the bath influence functional $\Ls_b$:
\begin{equation} \label{eq:L inchworm}
    \Ls_b^c(s_1,\cdots,s_{m},\sf)   =  \sum_{\mf{q} \in \mQ^c(\sb,\sf)} \prod_{(s_j,s_k) \in \mf{q}} B(s_j,s_k),
\end{equation}
but $\mQ^c$ is a subset of $\mQ$ appearing in $\Ls_b$ which only includes ``linked" pairings, which means in its diagrammatic representation any two points can be connected with each other using arcs as ``bridges". For example when $m=3$, $\Ls_b^c(s_1,s_2,s_3,\sf)$ only contains one linked diagram in \eqref{eq:all linking pair diagram example}:
\begin{equation} \label{Lbc m3}
\Ls_b^c(s_1,s_2,s_3,\sf)
= \ 
 \begin{tikzpicture}
\draw[-] (0,0)--(1.5,0);\draw plot[only marks,mark =*, mark options={color=black, scale=0.5}]coordinates {(0,0) (0.5,0) (1,0)(1.5,0)};
\draw[-] (0,0) to[bend left=75] (1,0);
\draw[-] (0.5,0) to[bend left=75] (1.5,0);
 \end{tikzpicture} = B(s_1,s_3) B(s_2,\sf).
\end{equation}
Another example for $m=5$ is given by 
\begin{equation}\label{linked pairs example}
  \begin{split}
    &\Ls_b^c(s_1,s_2,s_3,s_4,s_5,\sf)\\
  = & \ 
\begin{tikzpicture}
\draw[-] (0,0)--(2.5,0);\draw plot[only marks,mark =*, mark options={color=black, scale=0.5}]coordinates {(0,0) (0.5,0) (1,0)(1.5,0)(2,0)(2.5,0)};
\draw[-] (0,0) to[bend left=60] (1,0);
\draw[-] (0.5,0) to[bend left=60] (2,0);
\draw[-] (1.5,0) to[bend left=60] (2.5,0);
 \end{tikzpicture} 
 +
\begin{tikzpicture}
\draw[-] (0,0)--(2.5,0);\draw plot[only marks,mark =*, mark options={color=black, scale=0.5}]coordinates {(0,0) (0.5,0) (1,0)(1.5,0)(2,0)(2.5,0)};
\draw[-] (0,0) to[bend left=60] (1.5,0);
\draw[-] (0.5,0) to[bend left=60] (2,0);
\draw[-] (1,0) to[bend left=60] (2.5,0);
 \end{tikzpicture}  
 +
 \begin{tikzpicture}
\draw[-] (0,0)--(2.5,0);\draw plot[only marks,mark =*, mark options={color=black, scale=0.5}]coordinates {(0,0) (0.5,0) (1,0)(1.5,0)(2,0)(2.5,0)};
\draw[-] (0,0) to[bend left=60] (1.5,0);
\draw[-] (0.5,0) to[bend left=60] (2.5,0);
\draw[-] (1,0) to[bend left=60] (2,0);
 \end{tikzpicture} 
 +
  \begin{tikzpicture}
\draw[-] (0,0)--(2.5,0);\draw plot[only marks,mark =*, mark options={color=black, scale=0.5}]coordinates {(0,0) (0.5,0) (1,0)(1.5,0)(2,0)(2.5,0)};
\draw[-] (0,0) to[bend left=60] (2,0);
\draw[-] (0.5,0) to[bend left=60] (1.5,0);
\draw[-] (1,0) to[bend left=60] (2.5,0);
 \end{tikzpicture} \\
 := & \ B(s_1,s_3)B(s_2,s_5)B(s_4,\sf)+ B(s_1,s_4)B(s_2,s_5)B(s_3,\sf)\\
 & \hspace{60pt}  +B(s_1,s_4)B(s_2,\sf)B(s_3,s_5)+B(s_1,s_5)B(s_2,s_4)B(s_3,\sf) 
    \end{split}
\end{equation}
which does not include the unlinked terms in the bath influence functional $\Ls_b(s_1,s_2,s_3,s_4,s_5,\sf)$ such as 
\begin{equation} \label{unlinked diagrams} 
  \begin{split}
&\begin{tikzpicture}
\draw[-] (0,0)--(2.5,0);\draw plot[only marks,mark =*, mark options={color=black, scale=0.5}]coordinates {(0,0) (0.5,0) (1,0)(1.5,0)(2,0)(2.5,0)};
\draw[-,red] (0,0) to[bend left=75] (0.5,0);
\draw[-] (1,0) to[bend left=75] (2,0);
\draw[-] (1.5,0) to[bend left=75] (2.5,0);
 \end{tikzpicture}:= \redtext{B(s_1,s_2)}B(s_3,s_5)B(s_4,\sf),\\
 & \begin{tikzpicture}
\draw[-] (0,0)--(2.5,0);\draw plot[only marks,mark =*, mark options={color=black, scale=0.5}]coordinates {(0,0) (0.5,0) (1,0)(1.5,0)(2,0)(2.5,0)};
\draw[-] (0,0) to[bend left=75] (1,0);
\draw[-] (0.5,0) to[bend left=60] (2.5,0);
\draw[-,red] (1.5,0) to[bend left=75] (2,0);
 \end{tikzpicture} :=B(s_1,s_3)B(s_2,\sf)\redtext{B(s_4,s_5)} , \quad \cdots
  \end{split}
\end{equation}
where the pairs marked in red do not connect to the rest part of the diagrams via the arc bridges. Compared with the Dyson series, working with equation \eqref{eq: inchworm equation} is more advantageous as the series in the right-hand side has a faster convergence with respect to $m$. Also, $\Ls_b^c(s_1,\cdots,s_m,\sf)$ includes fewer diagrams than $\Ls_b(s_1,\cdots,s_m,t)$ in equation \eqref{dyson int diff eq 2} for the Dyson series, making its direct evaluation  cheaper than the bath influence functional for small $m$. However, asymptotically the number of diagrams in $\Ls_b^c(s_1,\cdots,s_m,\sf)$ also grows as a double factorial \cite{Stein1978b}, and its evaluation for large $m$ is even more expensive than $\Ls_b(s_1,\cdots,s_m,t)$ \cite{Yang2021}. Therefore, we again look for possible reuse of computed bath influence functionals when evolving the numerical solution.

\subsection{Numerical method}
Again, we truncate the series in the integro-differential equation up to a finite $\bar{M}$ as an approximation and apply the Runge-Kutta method 
for discretization on a uniform triangular mesh plotted in Figure \ref{fig:mesh}(a). For simplicity, we first consider the first-order forward Euler scheme:
\begin{equation}\label{inchworm RK}
 \begin{split}
  &\tilde{G}_{j,k} = \tilde{G}_{j,k-1}  + \sgn(t_{k-1} ) h \Bigg[
        \ii H_s \tilde{G}_{j,k-1}  +\\
    & \hspace{80pt}    + \sum^{\bar{M}}_{\substack{m=1\\ m \text{~is odd~}}} \ii^{m+1} \int_{t_j \le \sb \le t_{k-1}} \dd \sb   (-1)^{\#\{\sb < 0\}}  W_s \tilde{\mc{U}}_I(t_j,\sb,t_{k-1}) \Ls_b^c(\sb,t_{k-1})  \Bigg]
     \end{split}
\end{equation}
for $-N \le j < k \le N$ with $N = t_{\max}/h$. Here each $\tilde{G}_{j,k}$ is the approximation of the exact solution $G(t_j,t_k)$ and is denoted by a dot in Figure \ref{fig:mesh}(a). Since $\mc{U}$ defined in \eqref{UI} contains $G(s_k, s_{k-1})$ not on the grid points, we need to approximate $\mc{U}$ using $\tilde{\mc{U}}_I$ defined by
 \begin{equation}\label{UI}
\tilde{\Us}_I(t_j,s_1,\cdots,s_m,t_{k-1}) =  \tilde{G}_I(s_m,t_{k-1}) W_s \tilde{G}_I(s_{m-1}, s_{m}) W_s  \cdots  W_s \tilde{G}_I(s_1, s_2) W_s \tilde{G}_I(t_j, s_1)
\end{equation}
where the interpolating function $\tilde{G}_I(\cdot,\cdot)$ satisfies $\tilde{G}_I(t_\ell,t_n) = \tilde{G}_{\ell,n}$, for all $j \le \ell \le n \le k-1$, and the piecewise linear interpolation is adopted in our implementation. To ensure these $\tilde{G}_{\ell,n}$ are available before evaluating \eqref{UI}, we compute the full propagator column by column from left to right in Figure \ref{fig:mesh}(a). For each of these columns, we compute from top to bottom along the corresponding arrow. In order to better present the reuse of computed bath influence functionals, we consider the following decomposition of the domain of integration:
\begin{equation}\label{T split}
\{ \sb=(s_1,\cdots,s_m) \ | \   t_j \le s_1 \le \cdots \le s_m \le t_{k-1}\}  =   \bigcup_{p=j}^{k-2} T^{(m)}_{p,k-1} 
\end{equation}
where
\begin{equation}\label{Tpk}
 T^{(m)}_{p,k-1} = \{ (s_1, \cdots, s_m) \ |\ t_p \le s_1 \le t_{p+1}, s_1 \le s_2 \le \cdots \le s_m \le t_{k-1} \},
\end{equation}
which are pairwise disjoint for $p = j, \cdots, k-2$.
This decomposition can be visualized using the example in Figure \ref{fig:mesh}: when we use the scheme \eqref{inchworm RK} to evaluate $\tilde{G}_{-2,2}$ (node in red box in Figure \ref{fig:mesh}(a)), the domain of integration for $m=3$ is the simplex $\{(s_1, s_2, s_3) \mid -0.2 \le s_1  \le s_2 \le s_3 \le 0.1\}$ plotted in Figure \ref{fig:mesh}(b). According to the decomposition \eqref{T split}, this simplex can be split into $T^{(3)}_{0,1}$ (blue tetrahedron), $T^{(3)}_{-1,1}$ (red pentahedron) and $T^{(3)}_{-2,1}$ (green pentahedron). This decomposition allows us to reuse the bath integrand factor $\Ls_b^c(\sb,t_{k-1})$ when computing an integral in \eqref{inchworm RK} via 
\begin{equation}\label{bath reuse 1}
 \begin{split}
   & \int_{t_j \le \sb \le t_{k-1}} \dd \sb    (-1)^{\#\{\sb < 0\}}  W_s \tilde{\mc{U}}_I(t_j,\sb,t_{k-1}) \Ls_b^c(\sb,t_{k-1})\\
    & \hspace{80pt}  =    \int_{t_{j+1} \le \sb \le t_{k-1}} \dd \sb   (\text{integrand}) +  \int_{  \sb \in T^{(m)}_{j,k-1} } \dd \sb  (\text{integrand})  
      \end{split}
\end{equation}
where the value of $\Ls_b^c(\sb,t_{k-1})$ in the second integral above has been obtained when calculating $\tilde{G}_{j+1,k}$, while $\Ls_b^c(\sb,t_{k-1})$ in the last integral should be newly evaluated. This reuse of bath calculation can also be understood by the same example in Figure \ref{fig:mesh}: when evaluating $\tilde{G}_{-2,2}$, the values of $\Ls_b^c(\sb,0.1)$ for $\sb=(s_1,s_2,s_3)$ in $\{-0.1 \le s_1  \le s_2 \le s_3 \le 0.1\}$ (points in blue and red pentahedra) can be reused from $\tilde{G}_{-1,2}$ (node in blue box in Figure \ref{fig:mesh}(a)), and $\Ls_b^c(\sb,0.1)$ for $\sb \in T^{(2)}_{-2,1}$ (points in the green pentahedron) are to be calculated newly. However, we remark that such reuse does not apply to the entire integrand as the value of $\tilde{\mc{U}}_I$ replies on the $t_j$, which are different in $\tilde{G}_{-2,2}$ and $\tilde{G}_{-1,2}$.

\begin{figure}
    \centering
    \subfigure[]{\input{images/fig_mesh}}
     \subfigure[]{ \includegraphics[width=.45\textwidth]{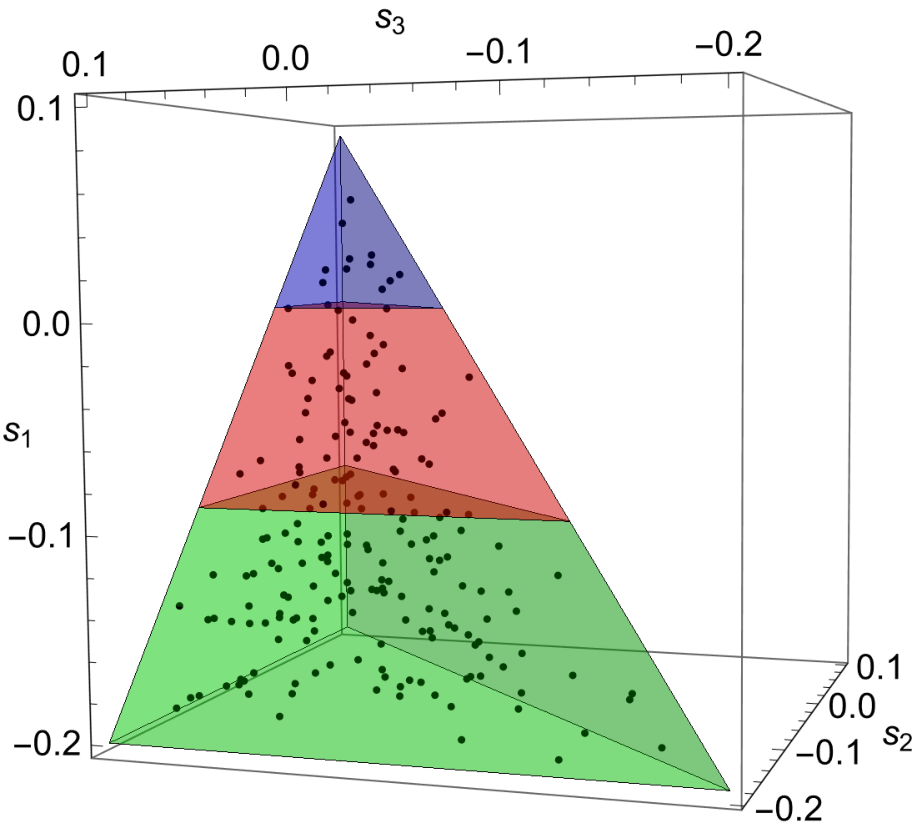}}
\caption{An example for $h = 0.1$, $N = 3$ and $t = 0.3$ (left: mesh structure, right: decomposition of the domain of integration $\{-0.2 \le s_1 \le s_2 \le s_3 \le 0.1\}$).}
 \label{fig:mesh}
\end{figure}

At this point, we draw time sequences $\Ss_{p,k-1} := \{ \sb^{(i)}_{p,k-1}\}_{i=1}^{{\Ms}_{p,k-1}}$ from $T_{p,k-1} = \bigcup_{ \substack{m=1\\ m \text{~ is odd}}}^{\bar{M}} T^{(m)}_{p,k-1}$ and approximate the sum of integrals in \eqref{inchworm RK} using Monte Carlo method. The numerical scheme becomes
\begin{multline*}
  G_{j,k} = G_{j,k-1}  + \sgn(t_{k-1} ) h \Bigg[
        \ii H_s G_{j,k-1}  +\frac{1}{\sum_{p=j}^{k-2} {\Ms}_{p,k-1}}\sum_{p=j}^{k-2} \  \sum_{i = 1}^{{\Ms}_{p,k-1}} \frac{1}{\P_{j,k-1}(m^{(i)}_{p,k-1}, \sb^{(i)}_{p,k-1}) }  \times \\
      \times \ii^{ m^{(i)}_{p,k-1}+1}  (-1)^{\#\{\sb^{(i)}_{p,k-1} < 0\}}   W_s \Us_I(t_j,\sb^{(i)}_{p,k-1},t_{k-1})\Ls_b^c(\sb^{(i)}_{p,k-1},t_{k-1}) \Bigg]
\end{multline*}
where the function $\P_{j,k-1}(m,\sb)$ gives the probability density of $(m,\sb)$ in $\bigcup_{p=j}^{k-2}T_{p,k-1}$, and the functional $\Us_I$ is similarly defined as \eqref{UI} with all $\tilde{G}_I$ replaced by $G_I$. The reuse of bath influence functionals stated in \eqref{bath reuse 1} is also reflected in the above scheme: when evaluating $G_{j,k}$, $\Ls_b^c(\sb^{(i)}_{p,k-1},t_{k-1})$ for $p =j+1,\cdots,k-2$ have already been obtained when computing $G_{j+1,k}$, and $\Ls_b^c(\sb^{(i)}_{p,k-1},t_{k-1})$ for $p=j$ should be newly calculated. Such implementation also indicates that one should follow a proper order to evolve the scheme, which will be discussed in detail in the next section.


To achieve a higher convergence order in time, we now put Heun's method \eqref{Heun} into this framework and the corresponding inchworm Monte Carlo method reads:
\begin{multline} \label{inchworm monte carlo}
   G_{j,k}^* = G_{j,k-1}  + \sgn(t_{k-1} ) h \Bigg[
        \ii H_s G_{j,k-1}  +\frac{1}{\sum_{p=j}^{k-2} {\Ms}_{p,k-1}}\sum_{p=j}^{k-2} \  \sum_{i = 1}^{{\Ms}_{p,k-1}} \frac{1}{\P_{j,k-1}(m^{(i)}_{p,k-1}, \sb^{(i)}_{p,k-1}) }  \times \\
      \times \ii^{ m^{(i)}_{p,k-1}+1}  (-1)^{\#\{\sb^{(i)}_{p,k-1} < 0\}}   W_s \Us_I(t_j,\sb^{(i)}_{p,k-1},t_{k-1})\Ls_b^c(\sb^{(i)}_{p,k-1},t_{k-1}) \Bigg],\\
 G_{j,k} =  \frac{1}{2} (G_{j,k-1} +  G_{j,k}^*) + \frac{1}{2} \sgn(t_{k} ) h \Bigg[
        \ii H_s G_{j,k}^*  +  \frac{1}{\sum_{p=j}^{k-1} {\Ms}_{p,k}}  \sum_{p=j}^{k-1} \ \sum_{i = 1}^{{\Ms}_{p,k}} \frac{1}{\P_{j,k}(m^{(i)}_{p,k}, \sb^{(i)}_{p,k}) } \times \\
     \times \ii^{(m^{(i)}_{p,k}+1)} (-1)^{\#\{\sb^{(i)}_{p,k} < 0\}} 
      W_s \Us^*_I(t_j,\sb^{(i)}_{p,k},t_{k})  \Ls_b^c(\sb^{(i)}_{p,k},t_{k})  \Bigg]
 \end{multline} 
where $\Us^*_I$ in the second stage is given by 
\begin{displaymath}
  \Us^*_I(t_j,s_1,\cdots,s_m,t_{k+1})  =  G_I^*(t_m, s_{k+1}) W_s G_I^*(s_{m-1}, s_{m}) W_s  \cdots  W_s G_I^*(s_1, s_2) W_s G_I^*(t_j, s_1) 
\end{displaymath}
with 
\begin{equation*}
G^*_I(t_\ell,t_n) = 
\begin{cases}
  G_{\ell,n}, & \text{if } (\ell,n) \neq (j,k+1), \\
  G^*_{j,k+1}, & \text{if } (\ell,n) = (j,k+1).
\end{cases}
\end{equation*}
In general, the inchworm Monte Carlo method \eqref{inchworm monte carlo} for the integro-differential equation \eqref{eq: inchworm equation} is similar to the scheme \eqref{dyson int diff eq scheme} for Dyson series, but for the inchworm method, some special care needs to be taken at time $t=0$, which will be detailed in the next section.

\subsubsection{General procedure of the inchworm Monte Carlo method}
\label{sec:time evo inchworm}
To apply the numerical scheme \eqref{inchworm monte carlo} accurately and efficiently, we need to take the properties of $G(\cdot,\cdot)$ into consideration, which leads us to the rules below that we should follow during the implementation:
\begin{itemize}
\item[(R1)] The evolution of the numerical scheme should begin with the boundary value $G_{j,j}=\text{Id}$ for $j = -N,\dots,N$, which are denoted by the red dots in Figure \ref{fig:reuse_order}(a).

\item[(R2)] Due to the discontinuities, when $j = 0$ or $k = 0$ (blue dots in Figure \ref{fig:mesh}(a)), $G_{j,k}$ is considered to be
multiple-valued, and we use $G_{0^{\pm},k}$ and $G_{j,0^{\pm}}$ respectively to represent the approximation of the left and right limits $\lim\limits_{s \rightarrow 0^{\pm}} G(s,
t_k)$ and $\lim\limits_{s \rightarrow 0^{\pm}} G(t_j, s)$. By the jump condition \eqref{jump condition}, we have the relation
\begin{displaymath}
 G_{j,0^+} = O_s G_{j,0^-}  \text{~and~} G_{0^-,k} =  G_{0^+,k}O_s \text{~for~}  -N\le j \le -1,1 \le k \le N. 
\end{displaymath}
In particular, the boundary value on the discontinuities are given by: $G_{0^+,0^+}=G_{0^-,0^-} =\text{Id}$ and $G_{0^-,0^+} = O_s$. Consequently, the interpolation of $G_I$ appearing in the functional $\Us_I$ should satisfy 
  \begin{gather*}
  \lim_{s\rightarrow 0^{\pm}}  G_I(t_j, s) = G_{j,0^{\pm}}, \qquad
  \lim_{s\rightarrow 0^{\pm}}  G_I(s, t_k) = G_{0^{\pm},k}, \\
  \lim_{s\rightarrow 0^+}
    \lim_{\tilde{s} \rightarrow 0^+}  G_I(\tilde{s}, s) = 
  \lim_{s\rightarrow 0^-}
    \lim_{\tilde{s} \rightarrow 0^-}  G_I(s, \tilde{s}) = \text{Id}, \qquad
  \lim_{s\rightarrow 0^-}
    \lim_{\tilde{s} \rightarrow 0^+}  G_I(s, \tilde{s}) = O_s
  \end{gather*}
and the conditions for $G_I^*$ in $\Us_I^*$ are similar. This rule is indispensable in our implementation to keep the second-order convergence rate in time of the Heun's method.

\item[(R3)] We only compute $G_{j,k}$ locating in the green triangle in Figure \ref{fig:reuse_order}(a) excluding the origin. Afterwards, the value of $G_{-k,-j}$ is obtained by the conjugate symmetry \eqref{conjugate symmetry}. In addition, the following full propagators are assigned with the same values according to the shift invariance \eqref{shift invariance}:  
\begin{equation}\label{discrete shift invariance}
 \begin{split}
  & G_{-N,-N+j} = G_{-N+1,-N+j+1} = \cdots =  G_{-j-1,-1} =    G_{-j,0^-} ,  \\
  & G_{N-j,N} = G_{N-j-1,N-1} = \cdots = G_{1,j+1} = G_{0^+,j}  \text{~for~}  1 \le j \le N-1.
   \end{split}
\end{equation}
\end{itemize}

\begin{algorithm}[ht] \small
  \caption{Evolution of inchworm Monte Carlo method}\label{algo:inchworm evo}
  \begin{algorithmic}[1]
\State \textbf{input}  $\Ss_{p,k}= \{\sb^{(i)}_{p,k}\}_{i=1}^{\Ms_{p,k}} \subset T_{p,k}$ and $L_{p,k} = \{\Ls_b^c(\sb^{(i)}_{p,k},t_k)\}_{i=1}^{\Ms_{p,k}}$ for $-N \le p \le k-1$ and $0\le k \le N$
    \State Set $G_{j,j}  \gets \text{Id}$ for $j = -N,\cdots,-1,0^{-},0^+,1,\cdots,N$ and $G_{0^-,0^+} \gets O_s$  \Comment{\emph{Initial condition}}
      \medskip
  \For{$n$ from $1$ to $N$}  \Comment{\emph{Time evolution on $n$th thick segment in Figure \ref{fig:reuse_order}(a)}} 
  \medskip
  \State Compute $G_{-n,0^-}$ by \eqref{inchworm monte carlo} and set $G_{-n,0^+} \gets O_s G_{-n,0^-}$, $G_{0^\pm,n} \gets G_{-n,0^\mp}$  \Comment{\emph{Compute blue dots}}  
    \medskip
    \For{$\ell_1$ from $1$ to $n$} \Comment{\emph{Compute inside quadrant IV}}   \State Compute $G_{-n,\ell_1}$ by \eqref{inchworm monte carlo} and set $G_{-\ell_1,n} \gets (G_{-n,\ell_1})^\dagger$
    \EndFor
    \medskip
     \For{$\ell_2$ from $1$ to $N-n-1$} \Comment{\emph{Assign values in quadrant I and III}}
    \State Set $G_{-n-\ell_2,-\ell_2} \gets G_{-n,0^-}$ and $G_{\ell_2,n+\ell_2} \gets G_{0^+,n}$
    \EndFor
   \EndFor
  \medskip
  \State \textbf{return} $G_{j,k}$ for $-N\le j \le k \le N$
  \end{algorithmic}
\end{algorithm}

We are now ready to sketch the implementation of the numerical scheme \eqref{inchworm monte carlo} in Algorithm \ref{algo:inchworm evo}. For the example in Figure \ref{fig:reuse_order}(a), the computation of full propagators should be first carried out on the thick red segment, followed by the blue thick segment and finally the black one. Such order of calculation is to guarantee that the values of shorter propagation $G_{n,\ell}$ for $j \le  n \le  \ell \le k$ are available before computing $G_{j,k}$ as argued previously. On each of these segments, we only evaluate $G_{j,k}$ in the green triangle from left to right, and the rest dots on the same segment can be directly assigned according to (R3). With the evolution of numerical scheme clarified, we now focus on the efficient construction of the input $\Ss_{p,k}$ and $L_{p,k}$, which will be specified in the next section.

\subsubsection{Time sequence sampling and reuse of bath calculation}
\label{sec:sampling inchworm}
Due to the invariance and symmetry of the full propagators, one only needs to evaluate $G_{j,k}$ in the green triangular area in Figure \ref{fig:reuse_order}(a) where the index $(j,k)$ satisfies $0 \le k \le -j$ and $-N \le j \le -1$. Since computing $G_{j,k}$ requires samples in $T_{p,k-1}$ for $j \le p \le k-2$ and $T_{p,k}$ for $j \le p \le k-1$ (see \eqref{inchworm monte carlo}), we need to prepare time sequences $\Ss_{p,k}$ and calculate $L_{p,k}$ for $-N \le p \le k-1$ and $0\le k \le N$ according to the splitting \eqref{T split}. These indices $(p,k)$ are denoted by the green nodes (both $``\times"$ and $``\bullet"$) in Figure \ref{fig:reuse_order}(b). In fact, we can focus only on those $``\bullet"$ nodes
since the set of samples $\Ss_{p,k}$ for $(p,k)$ on $``\times"$ nodes can be obtained by shifting the samples in $\Ss_{p-k,0}$:
\begin{displaymath}
 \sb^{(i)}_{p,k} =  \sb^{(i)}_{p-k,0} + t_{k} \, \b1  \text{~for~} 0 \le p < k \le N
\end{displaymath}
where $\b1$ is the row vector with all its components being $1$. We can then use \eqref{prop: B conj} to directly assign the bath value for these $``\times"$ nodes: 
\begin{displaymath}
 \Ls_b^c( \sb^{(i)}_{p,k},t_k ) = \overline{\Ls_b^c( \sb^{(i)}_{p-k,0},0 ) }. 
\end{displaymath}
Note that this cannot be applied to the green $``\bullet"$ nodes as the corresponding $(\sb^{(i)}_{p,k},t_k)$ contains both positive and negative entries and thus the condition of \eqref{prop: B conj} cannot be satisfied.

\input{images/fig_mesh_2}

From here, we will only work with $\Ss_{p,k}$ with $``\bullet"$ indices and consider the reuse of bath influence functionals for inchworm Monte Carlo method. Since the full propagator is now defined in the two-dimensional half-space, we have multiple paths  following which the bath influence functionals can be reused and each of these paths is represented by an arrow in the quadrant IV of Figure \ref{fig:reuse_order}(b). For example, the red arrow denotes the reuse of $\Ls_b^c(s_1,\cdots,s_m,t_k)$ for calculating $G_{-1,0} \rightarrow G_{-2,1} \rightarrow G_{-3,2}$, which is illustrated in Figure \ref{fig:reuse inchworm}. When $m=3$, each $\Ls_b^c(s_1,s_2,s_3,t_k)$ is represented by a linked diagram as defined in \eqref{Lbc m3}, where the time sequence $(s_1,s_2,s_3)$ denoted by the three dots is a sample in $T^{(3)}_{p,k}$. Note that in each diagram, the leftmost dot marked in red should always be restricted within $[t_p,t_{p+1}]$ by the definition of $T^{(m)}_{p,k}$. One can easily see that Proposition \ref{prop: L} also applies to $\Ls_b^c$. Hence, by the stretching invariance, the diagrams with the same color have the same functional value. In addition, we remark that the shifting invariance such as the reuse of red arc in Figure \ref{fig:reuse} is no longer considered now since the left end $t_p$ and right end $t_k$ satisfy $t_p < 0 \le t_k$ for all $``\bullet"$ nodes in Figure \ref{fig:reuse_order}(b). 

\input{images/property_L_inchworm}

Similar to algorithm for Dyson series, we consider the sample space $\sT_{p,k}$ from which we draw new time sequences in each time step (e.g, the blue diagram in $(-h,0)$, the red diagram in $(-2h,h)$ and the green diagram in $(-3h,2h)$). To do this, we generalize \eqref{T hat} to two indices: 
\begin{equation}\label{Tpk hat}
{\sT}^{(m)}_{p,k}  = \left\{ (s_1,\cdots, s_m) \in T^{(m)}_{p,k} \ \big| \ -h < t_k  < h \text{~or~}  \exists \ s_i \text{~such that~} -h < s_i < h  \right\},
\end{equation} 
where $p = -N, \cdots, -1$ and $k = 0, \cdots, N$. The volume of ${\sT}^{(m)}_{p,k}$ is similarly given by the relation: 
\begin{equation}\label{Tpk star vol}
   |{\sT}^{(m)}_{p,k} | =  
   \begin{cases}
  |T^{(m)}_{p,k}| -    |T^{(m)}_{p+1,k-1}| , & \text{~if~} -N\le p \le -2, 1\le k \le N, \\
    |T^{(m)}_{p,k}|,& \text{~if~} p = -1 \text{~or~} k = 0,
  \end{cases}
\end{equation}
where $|T^{(m)}_{p,k}|$ can be calculated by the definition \eqref{Tpk}: 
\begin{equation}\label{Tpk vol}
|T^{(m)}_{p,k} | = \int_{t_p}^{t_{p+1}} \dd s_1 \int_{s_1 \le s_2 \le \cdots \le s_m} \dd s_2 \cdots \dd s_m = \frac{1}{m!}[(t_{k-p})^{m} - ( t_{k-p-1})^m].
\end{equation}

In general, our algorithm for inchworm Monte Carlo method is summarized in Algorithm \ref{algo:inchworm reuse}. Lines 2--9 build up the time sequences $\Ss_{p,k}$ and bath influence functionals $L_{p,k}$ along the arrows in Figure \ref{fig:reuse_order}(b) following the strategy similar to Algorithm \ref{algo:dyson} for the Dyson series, and Lines 10--14 construct $\Ss_{p,k}$ and $L_{p,k}$ with $``\times"$ indices. The sampling method for the input $\hat{\Ss}_{p,k}$ is also similar to that for the Dyson series: the number of samples in $\hat{\Ss}_{p,k}$ is set as 
\begin{equation}\label{hat Mpk}
 \hat{\Ms}_{p,k}^{(m)} =  \frac{\hat{\Ms}_0}{\hat{\lambda}} \cdot |\sT^{(m)}_{p,k}| \cdot m!! \mathcal{B}^{\frac{m+1}{2}}
 \end{equation}
where $\hat{\Ms}_0 = \hat{\Ms}^{(1)}_{-1,0}$ and $\hat{\lambda} = \mathcal{B}h$. Each sample is again generated according to the uniform distribution $U(\sT^{(m)}_{p,k})$. The formula of the density function $\P_{j,k}(m,\sb)$ used in the numerical scheme is then given by the following proposition: 
\begin{proposition}
For any $-N \le j < k \le N$ and $m = 1,3,\cdots,\bar{M}$, 
\begin{equation}
\P_{j,k}(m,\sb) = \frac{1}{\lambda_{j,k}}   \cdot m!! \mathcal{B}^{\frac{m+1}{2}}
\end{equation}
where 
\begin{displaymath}
\lambda_{j,k} = \sum_{\substack{m'=1 \\ m' \text{~is odd}}}^{\bar{M}}  \frac{(t_{k-j})^{m'} }{(m'-1)!!} \cdot  \mathcal{B}^{\frac{m'+1}{2}}
\end{displaymath}
\end{proposition}
The proof of this proposition is almost identical to that of Proposition \ref{thm:P} and thus omitted.

\begin{algorithm}[ht] \small
  \caption{Efficient implementation of inchworm Monte Carlo method}\label{algo:inchworm reuse}
  \begin{algorithmic}[1]
\State \textbf{input}  $\hat{\Ss}_{p,k}= \{\sb^{(i)}_{p,k}\}_{i=1}^{\hat{\Ms}_{p,k}} \subset \sT_{p,k}$ for $-N \le p \le -1$ and $0\le k \le N$
 \medskip
  \For{$n$ from $1$ to $N$}  \Comment{\emph{Reuse on the arrows covering $(N-n+1)$ $``\bullet"$ in Figure \ref{fig:reuse_order}(b)}} 
    \For{$\ell$ from $0$ to $N-n$}
  \State Compute $\hat{L}_{-1-\ell,n+\ell} = \{\Ls_b^c(\sb^{(j)}_{-1-\ell,n+\ell},t_{n+\ell})\}_{j=1}^{\hat{\Ms}_{-1-\ell,n+\ell}}$ \Comment{\emph{Arrows starting from $p=-1$}} 
  \State Compute $\hat{L}_{-n-\ell,\ell} = \{\Ls_b^c(\sb^{(j)}_{-n-\ell,\ell},t_{\ell})\}_{j=1}^{\hat{\Ms}_{-n-\ell,\ell}}$ \Comment{\emph{Arrows starting from $k=0$}} 
\medskip 
\State Set $\Ss_{-1-\ell,n+\ell} \gets \bigcup_{j=0}^{\ell} \Is_{j} (\hat{\Ss}_{-1,n})$;\quad  $\Ss_{-n-\ell,\ell} \gets \bigcup_{j=0}^{\ell} \Is_{j} (\hat{\Ss}_{-n,0})$ \Comment{\emph{Stretch samples}} 
\medskip 
\State Set $L_{-1-\ell,n+\ell} \gets \bigcup_{j=0}^{\ell} \hat{L}_{-1,n}$;\quad  $L_{-n-\ell,\ell} \gets \bigcup_{j=0}^{\ell}\hat{L}_{-n,0}$ \Comment{\emph{Reuse $\hat{L}$ values}} 
  \EndFor
  \EndFor
  \medskip
 \For{$n$ from $1$ to $N$}\Comment{\emph{Assign values of $``\times"$}}
 \For{$\ell$ from $1$ to $N-n+1$}
 \State Set $\Ss_{-n+\ell,\ell} \gets \{  \sb^{(j)}_{-n,0} + t_{\ell} \cdot \b1^{(m_{-n,0}^{(j)})}  \}_{j=1}^{\hat{\Ms}_{-n,0}}$ and $L_{-n+\ell,\ell} \gets \hat{L}_{-n,0}$
  \EndFor
  \EndFor
      \medskip
  \State \textbf{return} $\Ss_{p,k}$ and $L_{p,k}$ for $-N \le p \le k-1$ and $0\le k \le N$
  \end{algorithmic}
\end{algorithm}

\begin{remark}
 Unlike the method based on Dyson series, low memory cost implementation for inchworm method is not available. The system factor $\mathcal{U}_I^*$ in the numerical scheme \eqref{inchworm monte carlo} now depends on the the previously computed full propagators, which prohibits the preparation of all the partial sums like we did in Section \ref{sec:comp cost dqmc}. Consequently, these sums have to be computed sequentially, and all the bath influence functionals have to be stored to gain the efficiency. One possible workaround for long-time simulations is to restrict the memory length like in the iterative QuAPI method \cite{Makri1995}. We will leave this to future works. 
\end{remark}

\subsection{Analysis on computational cost}
 \label{sec:comp cost inchworm}
We again examine the computational cost saved after reusing the bath calculations for inchworm Monte Carlo method. By \eqref{hat Mpk}, the total number of samples in $\hat{\Ss}_{p,k}$ with $``\bullet"$ indices in Figure \ref{fig:reuse_order}(b) is in general given by 
\begin{equation}\label{num hat s inchworm}
 \#\{\ssb^{(m)}\} =  \frac{\hat{\Ms}_0}{\hat{\lambda}} \cdot \left(  \sum_{n=1}^{N} \sum_{\ell = 0}^{N-n} |\sT^{(m)}_{-1-\ell,n+\ell}| + |\sT^{(m)}_{-n-\ell,\ell}| \right) \cdot m!! \mathcal{B}^{\frac{m+1}{2}}
\end{equation}
where $|\sT^{(m)}_{p,k}|$ is summed along the arrows in the quadrant IV. Applying the relation \eqref{Tpk star vol} on each arrow yields  
\begin{displaymath} 
   \begin{split}
  \#\{\ssb^{(m)}\}   = & \  \frac{\hat{\Ms}_0}{\hat{\lambda}}\cdot \left(  \sum_{n=1}^N  |T^{(m)}_{-1-N+n,N}| + |T^{(m)}_{-N,N-n}| \right) \cdot m!! \mathcal{B}^{\frac{m+1}{2}} \\
  = & \  \frac{ \hat{\Ms}_{0}\mathcal{B}^{ \frac{m+1}{2}}}{\hat{\lambda}(m-1)!!} \cdot \left[ (t_{2N})^m + (t_{2N-1})^m - (t_{N})^m -(t_{N-1})^m\right].
   \end{split}
\end{displaymath}
One the other hand, similar to \eqref{num s dyson} for Dyson series, the number of all time sequences, denoted by $\#\{\sb^{(m)}\}$, is expressed by \eqref{num hat s inchworm} with the volume $|\sT^{(m)}_{p,k}|$ replaced by $|T^{(m)}_{p,k}|$. Note that the value of $|T^{(m)}_{p,k}|$ only depends on the difference $k-p$ according to \eqref{Tpk vol}, we therefore have
\begin{displaymath} 
   \begin{split}
  \#\{\sb^{(m)}\} = & \  \frac{\hat{\Ms}_0}{\hat{\lambda}}\cdot \left(  \sum_{i=1}^{2N} \  \sum_{\substack{ -N \le p \le -1, \\ 0 \le  k \le N, \\ k-p = i }} |T^{(m)}_{p,k}| \right) \cdot m!! \mathcal{B}^{\frac{m+1}{2}} \\
  = & \  \frac{\hat{\Ms}_0}{\hat{\lambda}} \cdot \left( \sum_{j=1}^N j \cdot \frac{(t_j)^m - (t_{j-1}^m)}{m!} + \sum_{j=N+1}^{2N} (2N-j+1)\cdot \frac{(t_j)^m - (t_{j-1}^m)}{m!} \right) \cdot m!! \mathcal{B}^{\frac{m+1}{2}} \\
  = & \ \frac{ \hat{\Ms}_{0}\mathcal{B}^{\frac{m+1}{2}}}{\hat{\lambda}(m-1)!!} \cdot \left( \sum_{j=N+1}^{2N} (t_j)^m - \sum_{j=1}^{N-1} (t_j)^m  \right).
   \end{split}
\end{displaymath}
Thus, for the order-$m$ bath influence functionals, the proportion of the computational cost saved by the reuse is 
\begin{equation}\label{save percent inchworm}
   R^{(m)}(N) = 1 - \frac{ \#\{\ssb^{(m)}\}}{ \#\{\sb^{(m)}\} } =  1 - \frac{(2N)^m + (2N-1)^m - (N)^m - (N-1)^m   }{ \sum_{j=N+1}^{2N} (j)^m - \sum_{j=1}^{N-1} (j)^m  }.
\end{equation}
For large $N$, the denominator can again be estimated by Faulhaber's formula \eqref{faulhaber}:
\begin{displaymath}
 \sum_{j=N+1}^{2N} (j)^m - \sum_{j=1}^{N-1} (j)^m =  \sum_{j=1}^{2N} (j)^m - 2  \sum_{j=1}^{N-1} (j)^m - N^m \sim \frac{(2N)^{m+1} - 2(N-1)^{m+1}}{m+1} - N^m,
\end{displaymath}
yielding the following asymptotic growth of $R^{(m)}$:
\begin{equation}\label{inchworm R asymptotic}
  R^{(m)}(N) \sim  1 - \frac{1- (\frac{1}{2})^{m+1}}{1- (\frac{1}{2})^m} \cdot \frac{m+1}{N},
\end{equation}
which also converges to $1$ at the rate $O(\frac{1}{N})$. It can be seen that this asymptotic value is close to $1 - (m+1)/N$ as in \eqref{dyson R asymptotic} for the Dyson series, especially for large $m$.
In particular, one can check that \eqref{save percent dqmc} and \eqref{save percent inchworm} are equal when $m=1$. This similarity can be verified by the graphs of $R^{(m)}$ in Figure \ref{fig:percent inchworm} where the dashed lines are almost identical to those in Figure \ref{fig:percent dqmc} for Dyson series, suggesting that inchworm Monte Carlo method can benefit the same reduction in the computational cost of $\Ls_b^c(s_1,\cdots,s_m)$ after reusing the bath calculations. By further taking the computational complexity of $\Ls_b^c(s_1,\cdots,s_m)$ into consideration, which is $O(\alpha^m)$ with $\alpha \approx 2.1258$ upon applying the inclusion-exclusion principle \cite[Section 3]{Yang2021}, the overall reduction of the computational cost can again be formulated as \eqref{RT} with $\mathcal{T}^{(m)}$ denoting the average wall clock time on evaluating $\Ls_b^c(s_1,\cdots,s_m,\sf)$, which has the following asymptotic behavior:
\begin{displaymath}
  \begin{split}
 R_{\text{T}} \sim   R^{\text{asy}}_{\text{T}}  = & \  1 - \frac{\sum^{\bar{M}}_{\substack{m = 1 \\ m \text{~is odd}}}  \#\{\ssb^{(m)}\} \cdot \alpha^{m}}{\sum^{\bar{M}}_{\substack{m = 1 \\ m \text{~is odd}}}  \#\{\sb^{(m)}\} \cdot \alpha^{m}} \\
 = & \ 1 - \frac{\sum_{m=1}^{\frac{\bar{M}+1}{2}}\frac{(\sqrt{\mathcal{B}}\alpha)^{2m-1}}{(2m-2)!!} \cdot \left[ (t_{2N})^{2m-1} + (t_{2N-1})^{2m-1}  - (t_{N})^{2m-1}  -(t_{N-1})^{2m-1} \right]}{\sum_{m=1}^{\frac{\bar{M}+1}{2}}\frac{(\sqrt{\mathcal{B}}\alpha)^{2m-1}}{(2m-2)!!} \cdot \left( \sum_{j=N+1}^{2N} (t_j)^{2m-1} - \sum_{j=1}^{N-1} (t_j)^{2m-1}  \right)} \\
 \sim & \  1 - \frac{(\bar{M}+1)\left(2- (\frac{1}{2})^{\bar{M}})\right)}{2\left(1- (\frac{1}{2})^{\bar{M}})\right) } \cdot \frac{1}{N} \text{ as } N \rightarrow + \infty.
  \end{split}
\end{displaymath}
This ratio again converges to $R^{(\bar{M})}$ as shown by the solid lines in Figure \ref{fig:percent inchworm}.

\begin{figure}[ht]
  \includegraphics[scale=0.5]{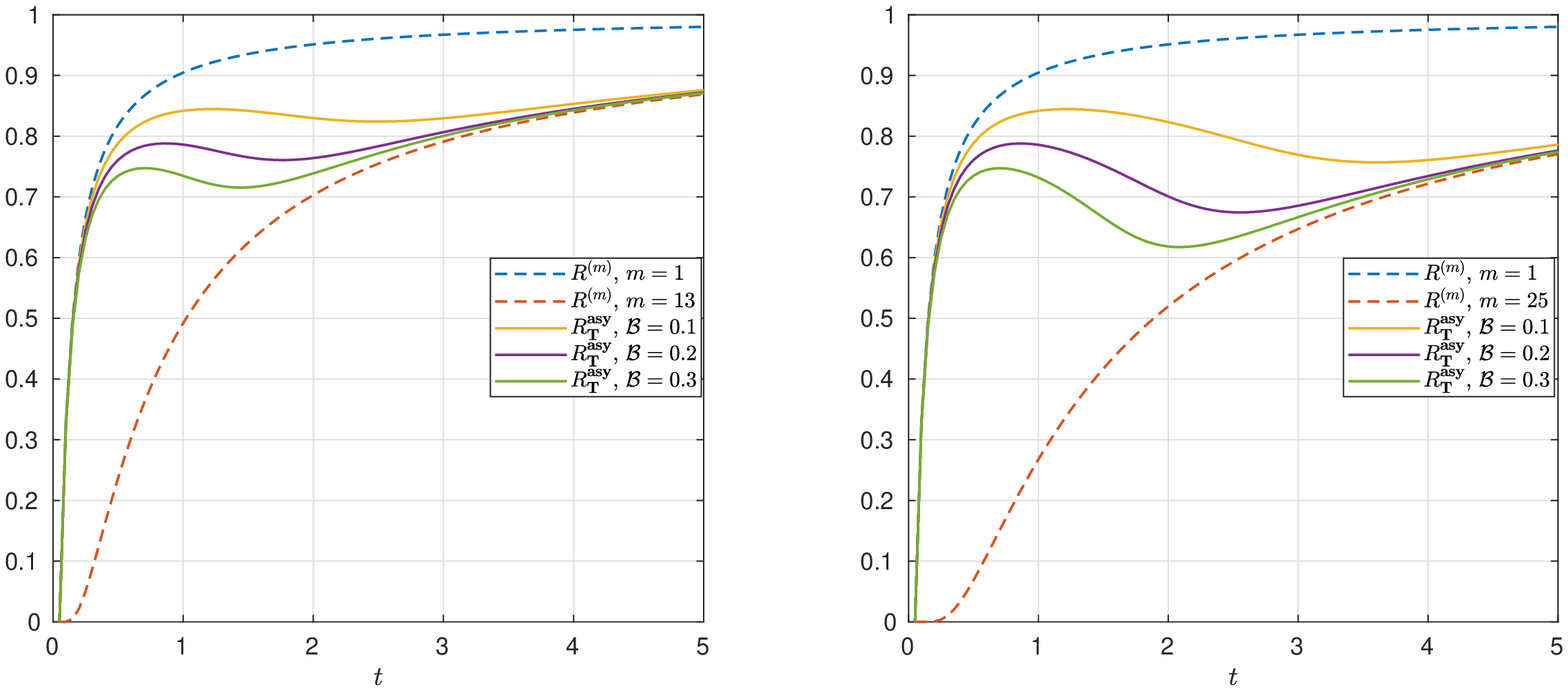}
  \caption{Graphs of $R^{(m)}$ and $R^{\text{asy}}_{\text{T}}$ for inchworm Monte Carlo method (left: $\bar{M}=13$, right: $\bar{M} = 25$).}
   \label{fig:percent inchworm}
\end{figure}

\section{Numerical experiments}
\label{sec:num exp}

In our numerical experiments, we consider the spin-boson model where the system Hamiltonian has the energy difference $\epsilon =1$ and frequency of the spin flipping $\Delta =1$. For the bath influence functional, we assume an Ohmic spectral density
\begin{displaymath}
 J(\omega) = \frac{\pi}{2} \sum_{l=1}^L \frac{c^2_{l}}{\omega_{l}} \delta (\omega - \omega_{l})
\end{displaymath}
where the number of modes is set as $L=400$. The coupling intensity $c_l$ and frequency of each harmonic oscillator $\omega_l$ above are respectively given by 
\begin{displaymath}
c_l = \omega_l \sqrt{\frac{\xi \omega_c}{L} [1 - \exp(-\omega_{\max}/\omega_c)]}, \quad \omega_l = -\omega_c
  \ln \left( 1 - \frac{l}{L} [1 - \exp(-\omega_{\max} / \omega_c)] \right),
  \ l = 1,\cdots,L
\end{displaymath}
where the maximum frequency is set as $\omega_{\max} = 4\omega_c$. Hence, the two-point correlation \eqref{def:B} is formulated as 
\begin{displaymath}
B(\tau_1, \tau_2)    = \sum_{l=1}^L \frac{c_l^2}{2\omega_l} \left[
 \coth \left( \frac{\beta \omega_l}{2} \right) \cos \big( \omega_l \Delta \tau \big)
 - \ii \sin\big( \omega_l \Delta \tau )
\right].
\end{displaymath} 

In Figure \ref{fig:bath cor}, we plot the amplitude of the two-point correlation with Kondo parameter $\xi=0.4$, inverse temperature $\beta=5$ and primary frequency $\omega_c = 2.5$ as the orange curve. The empirical constant $\mathcal{B}$ appearing in \eqref{P s* in T*} and \eqref{hat Mpk} should then be chosen between $(0,1.2)$. Larger $\mathcal{B}$ will lead to more time sequences sampled with large $m$, and thus a higher computational cost. In practice, one may start with a relatively small $\mathcal{B}$ to see whether the variance is small enough.
If not, one may then increase $\mathcal{B}$ and repeat the simulation. According to our tests, choosing $\mathcal{B}=0.2$ provides satisfactory results. We will also consider another numerical example with $\xi = 0.2$ for which the modulus of the two-point bath correlation is given by the blue curve in Figure \ref{fig:bath cor}. The corresponding $\mathcal{B}$ is set to be $0.1$. For all our numerical examples in this section, we truncate the series in \eqref{dyson int diff eq 2} and \eqref{eq: inchworm equation} by $\bar{M}=11$.

\begin{figure}[ht]
  \includegraphics[scale=0.5]{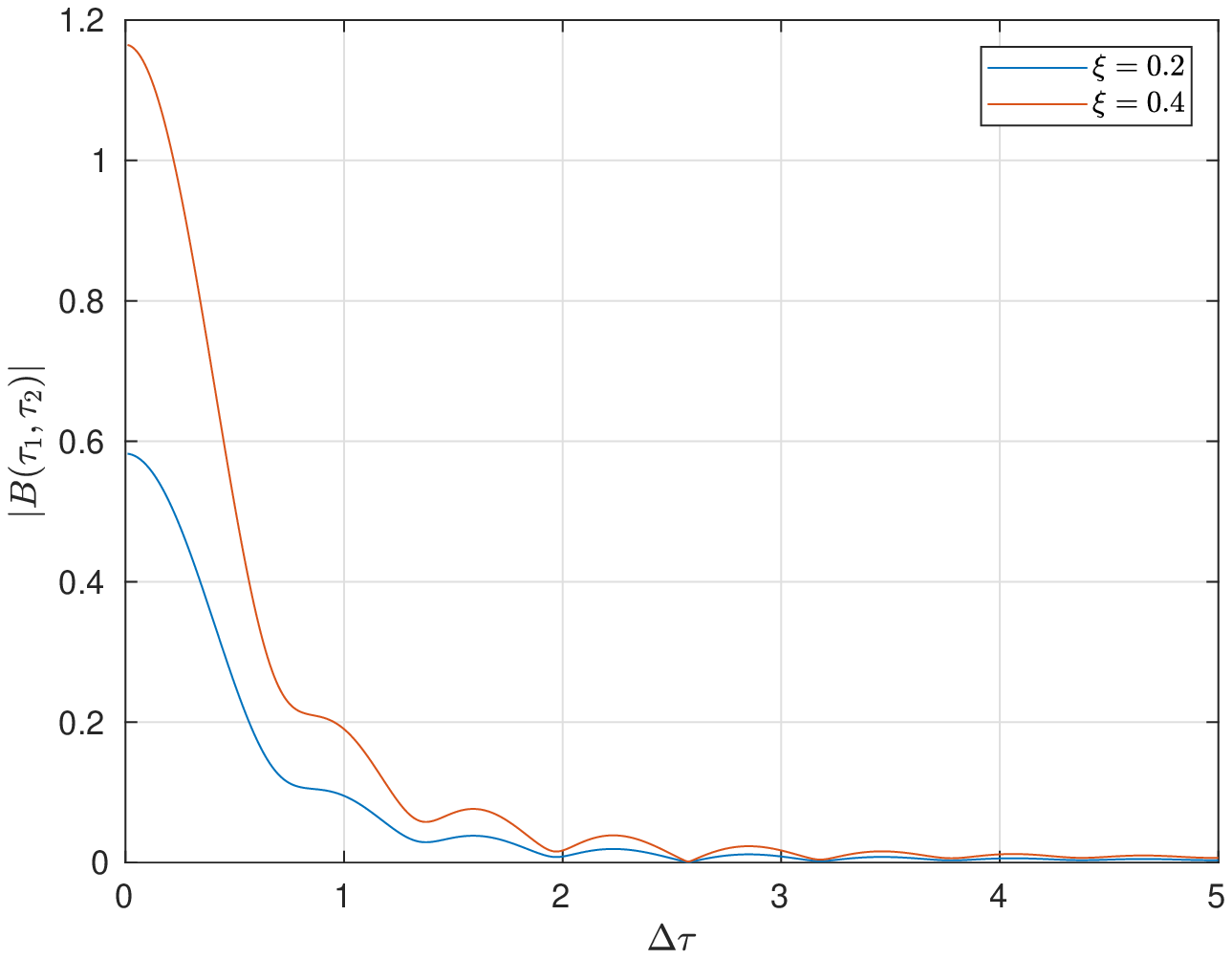}
    \caption{Two-point correlation functions for different Kondo parameters (orange: $\xi = 0.4$, blue: $\xi = 0.2$).}
    \label{fig:bath cor}
\end{figure}

\subsection{Evolution of observable}
To validate our numerical method, we first apply our reuse of bath calculations to both coupling intensities and compare the evolution of the observables to the results of classical methods.
The observable of interest is set to be $O = \hat{\sigma}_z \otimes \mathrm{Id}_b$ which only acts on the system, and the initial density matrix $\rho = \rho_s \otimes \rho_b$ is given by 
\begin{displaymath}
   \rho_s = \ket{0} \bra{0} \quad  \text{~and~} \quad  \rho_b = Z^{-1} \exp(-\beta H_b)\,,
\end{displaymath}
where $Z$ is a normalizing factor satisfying $\tr(\rho_b) = 1$. The evolution of observable $\langle\hat{\sigma}_z(t)\rangle$ is then evaluated discretely by 
\begin{displaymath}
  \langle\hat{\sigma}_z(nh)\rangle \approx \bra{0} G_{-n,n} \ket{0} \text{~for~} n = 0,1,\cdots,N
\end{displaymath}
where $G_{-n,n}$ is computed by either scheme \eqref{dyson int diff eq scheme} for Dyson series or scheme \eqref{inchworm monte carlo} for inchworm Monte Carlo method. 

In our numerical tests, we set the time step to be $h = 0.05$. It is generally believed that the inchworm Monte Carlo method requires less samples than the summation of the Dyson series. Therefore we set the initial number of samples $\hat{\mathcal{M}}_{0}$ to be $10^5$ for the solver of the Dyson equation \eqref{dyson int diff eq 2}, and set  $\hat{\mathcal{M}}_{0} = 10^4$ for the inchworm Monte Carlo method.
In Figure \ref{fig:observable}, we plot the numerical results of observable for both Kondo parameters. The results by iterative QuAPI method \cite{Makri1995,Makri1998} are also given as the reference solutions. In the left panel, the two curves are hardly distinguishable and both match the reference solution well. In the right panel, however, an obvious difference between two curves can be observed after $t=2.5$ and the result of the iterative QuAPI method indicates that the inchworm Monte Carlo method gives a better approximation. This is due to the fact that the larger amplitude of $B(\tau_1,\tau_2)$ with $\xi=0.4$ makes the Dyson series harder to converge with respect to $m$ for long time simulations. As a result, the truncation $\bar{M}=11$ is no longer sufficient for Dyson series, but still works for the inchworm Monte Carlo method thanks to its faster convergence as mentioned in Section \ref{sec:inchworm int diff eq}.

\begin{figure}[ht]
  \includegraphics[scale=0.5]{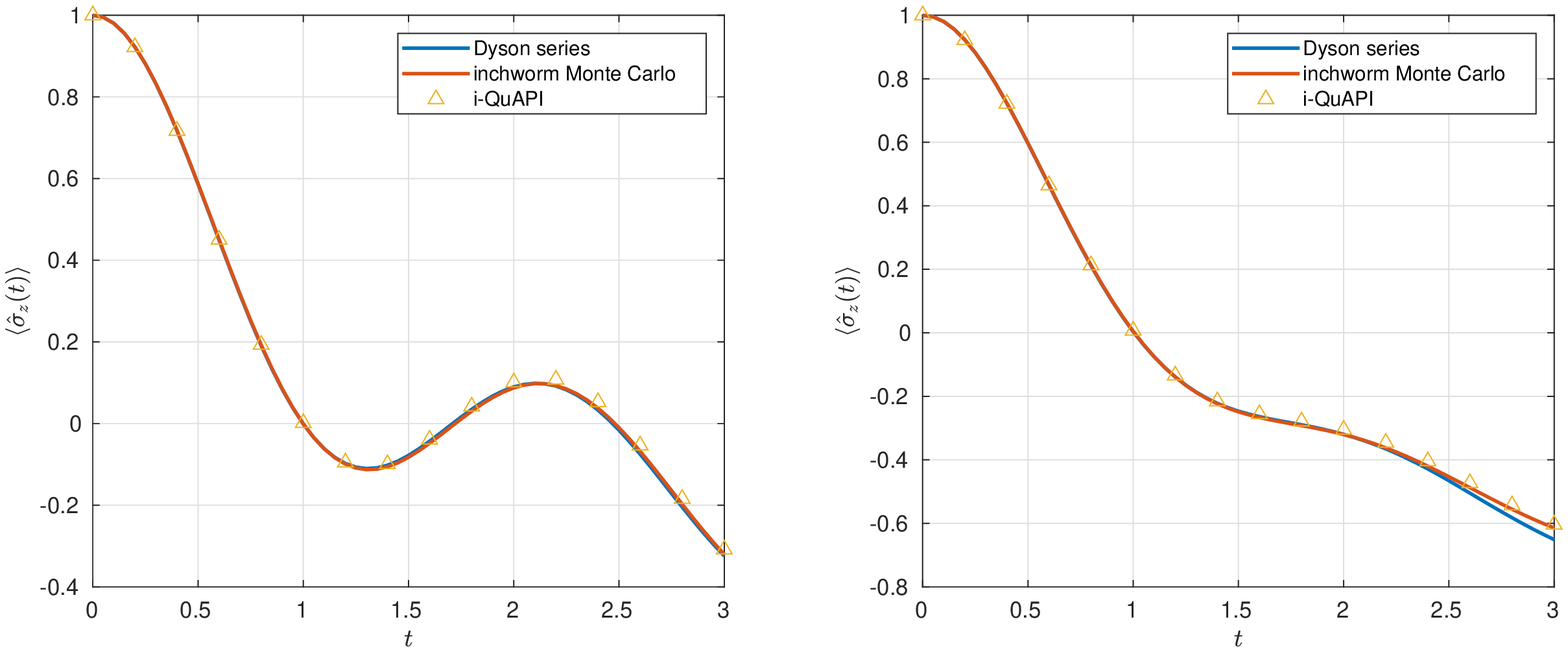}
    \caption{Evolution of $\langle \hat{\sigma}_z(t) \rangle$ under different settings of
the Kondo parameter (left: $\xi = 0.2$, right: $\xi = 0.4$).}
    \label{fig:observable}
\end{figure}

\subsection{Accuracy test}
\label{sec:accuracy}

To verify the accuracy of the numerical discretization by Heun's method used throughout this paper, we plot the results of $\langle \hat{\sigma}_z(t) \rangle$ computed by both algorithms with different time steps in Figure \ref{fig_accuracy}. The parameters of simulations are set to be the same as the left panel of Figure \ref{fig:observable}. For Dyson series, the result of $h = 0.05$ is indistinguishable with the result of $h = 0.025$ by naked eyes, while the curve for $h = 0.1$ still shows observable discrepancy with the other two lines. Note that $h = 0.05$ is used for the simulations in Figure \ref{fig:observable}, which is now proven to be reliable according to our accuracy test. For the inchworm Monte Carlo method, the convergence is achieved at a coarser grid $h = 0.1$, which is possibly due to the smaller number of terms in the bath influcence functional.
As a comparison, we also plot the results by first-order Forward Euler scheme (dashed curves), which obviously have not converged at $h=0.05$. This shows the advantage of Heun's method in terms of the accuracy of time discretization. While the second-order Heun's scheme is sufficient to produce accurate simulations up to $t=3$ in the current work, it is also worthwhile to consider higher-order or implicit schemes for the integral-differential equations \eqref{dyson int diff eq 2} and \eqref{eq: inchworm equation} to achieve better accuracy and stability.

\begin{figure}[ht]
  \includegraphics[scale=0.5]{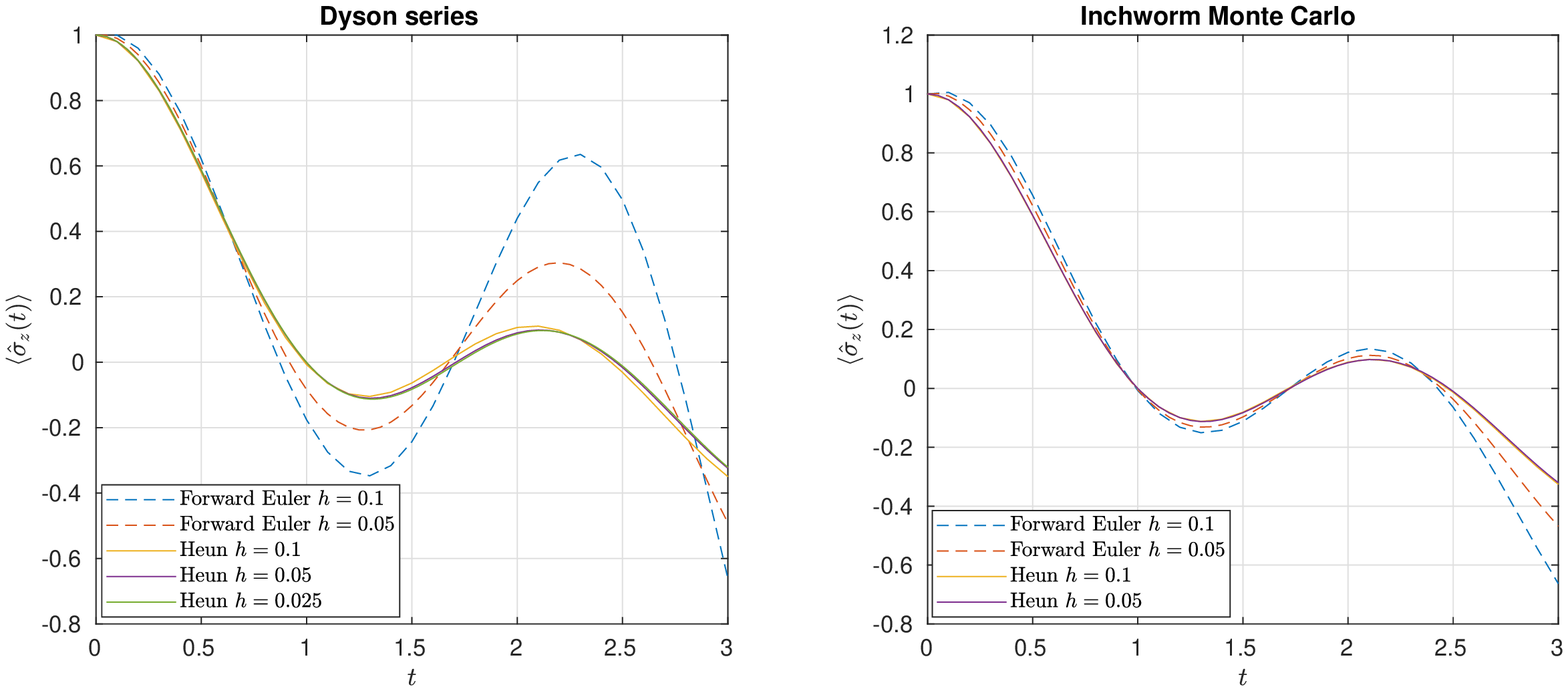}
    \caption{Evolution of $\langle \hat{\sigma}_z(t) \rangle$ for various time step lengths.}
    \label{fig_accuracy}
\end{figure}

\subsection{Efficiency test}
\label{sec:time complexity}
We now examine the computational time that can be saved by reusing the bath calculations. The experiments are carried out using MATLAB on AMD Ryzen 7 4800H CPU, and we use the parameters for the orange curve ($\xi = 0.4$) in Figure \ref{fig:bath cor} for the efficiency tests.   

We first compare the wall clock time on evaluating a given system associated $\mc{U}^{(0)}$ with that on $\Ls_b$ appearing in the integrand of Dyson series. As shown in Table \ref{tab:U vs L time}, the evaluation of $\Ls_b$ is more expensive than $\mc{U}^{(0)}$ in terms of time consumed for all choices of $m$. As $m$ increases, this difference becomes larger due to the linear complexity of $\mc{U}^{(0)}$ and exponential complexity of $\Ls_b$. Therefore, the computational cost on the bath influence functional dominates the overall evaluation of a given Dyson series. Similar conclusion for the inchworm Monte Carlo method can be drawn by Table \ref{tab:U vs L time inchworm}, where we list the wall clock time of $\mc{U}_I$ and $\Ls_b^c$ in the scheme \eqref{inchworm monte carlo}. Here both $\Ls_b$ and $\Ls_b^c$ are computed using the fast algorithms based on inclusion-exclusion principle as mentioned previously. Instead of directly summing the linked diagrams in \eqref{linked pairs example}, a given $\Ls_b^c(\sb,t)$ is evaluated indirectly under such algorithms which relies on the value of $\Ls_b(\sb,t)$ as well as $\Ls_b(\tilde{\sb},t)$ for some subsequences $\tilde{\sb} \subset \sb$, making $\Ls_b^c$ in general more expensive than $\Ls_b$ despite the fact that $\Ls_b^c$ contains fewer diagrams. We refer the readers to \cite{Yang2021} for more details of the algorithm. On the other hand, the computation of $\mc{U}_I$ defined by \eqref{UI} in inchworm method is faster than $\mc{U}^{(0)}$ defined by \eqref{def U0} in Dyson series since each matrix $G_I$ in $\mc{U}_I$ is obtained by linear interpolation, which is cheaper than $G^{(0)}_s$ in $\mc{U}^{(0)}$ where a matrix exponential is to be computed.

   \begin{table}[ht]
  \centering 
  \begin{tabular}{|c| c c c c c c|}
\hline
  $m$ & 1 & 3 & 5 & 7 & 9 & 11 \\
\hline
  $\mc{U}^{(0)}$  & 6.8000e-05  &    1.1800e-04 &    1.8000e-04 &   2.0800e-04 & 2.3200e-04 & 3.6500e-04\\
\hline
 $\Ls_b$  &  1.0100e-04  &    3.2800e-04 & 6.7200e-04 &   0.0011 & 0.0016 &  0.0023 \\
 \hline 
\end{tabular}
  \caption{Wall clock time (seconds) on evaluating a given $\mc{U}^{(0)}(-t,s_1,\cdots,s_m,t)$ and $\Ls_b(s_1,\cdots,s_m,t)$.}
   \label{tab:U vs L time}
\end{table}

\begin{table}[ht]
  \centering 
  \begin{tabular}{|c| c c c c c c|}
\hline
  $m$ & 1 & 3 & 5 & 7 & 9 & 11 \\
\hline
  $\mc{U}_I$  & 2.8000e-05  &    5.3000e-05 &    7.2000e-05 &   1.2600e-04 & 1.5900e-04 & 1.7000e-04\\
\hline
 $\Ls_b^c$  &  8.8000e-05  &    4.0200e-04 & 0.0010 &   0.0025 & 0.0053 &  0.0118 \\
 \hline 
\end{tabular}
  \caption{Wall clock time (seconds) on evaluating a given $\mc{U}_I(\si,s_1,\cdots,s_m,\sf)$ and $\Ls_b^c(s_1,\cdots,s_m,\sf)$.}
   \label{tab:U vs L time inchworm}
\end{table}

In Figure \ref{fig:time save}, we plot the theoretical savings in computational time spent on bath computations $R_{\mathrm{T}}$ defined by \eqref{RT} as the yellow solid lines, where the average wall clock time $\mathcal{T}^{(m)}$ for $\Ls_b$ in Dyson series and $\Ls_b^c$ in inchworm Monte Carlo method are respectively assigned with the values in Table \ref{tab:U vs L time} and \ref{tab:U vs L time inchworm}. The graphs of $R^{(1)}$ and $R^{(11)}$ are also plotted as the reference. As augured in Section \ref{sec:comp cost dqmc} and \ref{sec:comp cost inchworm}, $R_{\mathrm{T}}$ is always bounded by $R^{(1)}$ and $R^{(11)}$.

Meanwhile, we carry out two sets of numerical simulations under both methods with the initial number of samples $\hat{\mathcal{M}}_0 = 100$. In the first set of simulations, we apply the bath calculation reuse and record the total time spent on the bath influence functional up to $n$th time step as $\hat{\mathcal{T}}(n)$, while the second set are implemented without reusing calculations and the time consumed on bath is denoted by $\mathcal{T}(n)$. Then we may use the ratio 
\begin{displaymath}
  R_{\mathrm{T}}^{\text{real}}(n) = 1 - \hat{\mathcal{T}}(n)/\mathcal{T}(n)
\end{displaymath}
to measure the overall saving in time in real implementations, which are plotted as the purple solid lines in Figure \ref{fig:time save}. Since each evaluation on $\Ls_b$ or $\Ls_b^c$ cannot cost exactly the same amount of time, some oscillations can be observed in the purple curves. Nevertheless, $R_{\mathrm{T}}^{\text{real}}$ generally matches the theoretical $R_{\mathrm{T}}$ as $t$ grows, and thus we have verified the complexity analysis in Section \ref{sec:comp cost dqmc} and \ref{sec:comp cost inchworm}. As time further evolves, we may expect the overall saving in time to gradually converge to $R^{(11)}$ (orange dashed lines). Therefore, asymptotically the bath calculation reuse can achieve a total reduce in computational time at around the percentage $1- \frac{12}{n}$ for both Dyson series and inchworm Monte Carlo method for this example according to \eqref{dyson R asymptotic} and \eqref{inchworm R asymptotic}.

\begin{figure}[ht]
  \includegraphics[scale=0.5]{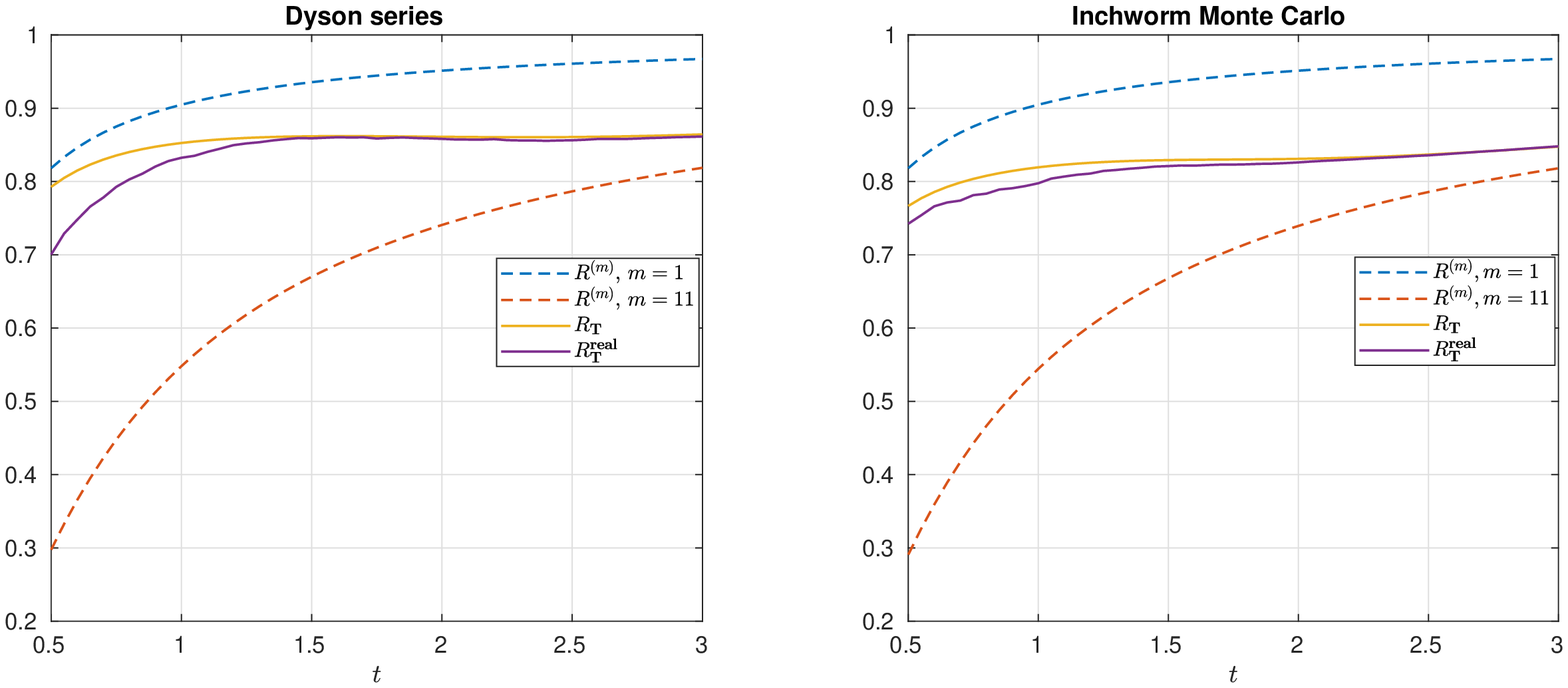}
  \caption{Overall savings of computational time spent on bath calculations.}
   \label{fig:time save}
\end{figure}

\subsection{Order of convergence}
As both numerical methods we have developed are stochastic schemes based on Monte Carlo, it is of interest to study the convergence rate of the standard derivation of the numerical solution with respect to the initial number of samples $\hat{\mathcal{M}}_0$. In this experiment, we fix the time step length as $h=0.1$ and compute up to $t=1$. The parameter setting for the two-point correlation is given as $\omega_c = 1$, $\xi=0.1$ and $\beta = 0.2$ with the empirical constant $\mathcal{B}=0.3$. We run the same simulation independently for $N_{\exp}=1000$ times, and the standard derivation of $G_{-n,n}$ is estimated as  
\begin{displaymath}
 \sigma_{\hat{\mathcal{M}}_0}(t_n) = \left( \frac{1}{N_{\exp}} \sum^{N_{\exp}}_{k=1} \left\| G^{[k]}_{-n,n} - \mu_{-n,n}  \right\|_{\mathrm{F}}^2  \right)^{1/2}, \text{~for~} n = 0,1,\cdots,20
\end{displaymath}
where $\|\cdot\|_{\mathrm{F}}$ denotes the Frobenius norm. Here $G^{[k]}_{-n,n}$ is the result of $k$th numerical simulation, and $\mu_{-n,n}$ should be the expectation of $G_{-n,n}$, which in our implementation is replaced by the numerical exact solution $G_{-n,n}$ that is computed based on a large initial number of samples $\hat{\mathcal{M}}_0 = 10^6$ for Dyson series and $\hat{\mathcal{M}}_0 = 10^5$ for inchworm Monte Carlo method. The numerical results are shown in Figure \ref{fig:conv order}, where the $1/2$ order of convergence for the standard derivation is obvious, indicating that the optimal convergence rate of Monte Carlo method is achieved in both stochastic schemes.

\begin{figure}[ht]
  \includegraphics[scale=0.5]{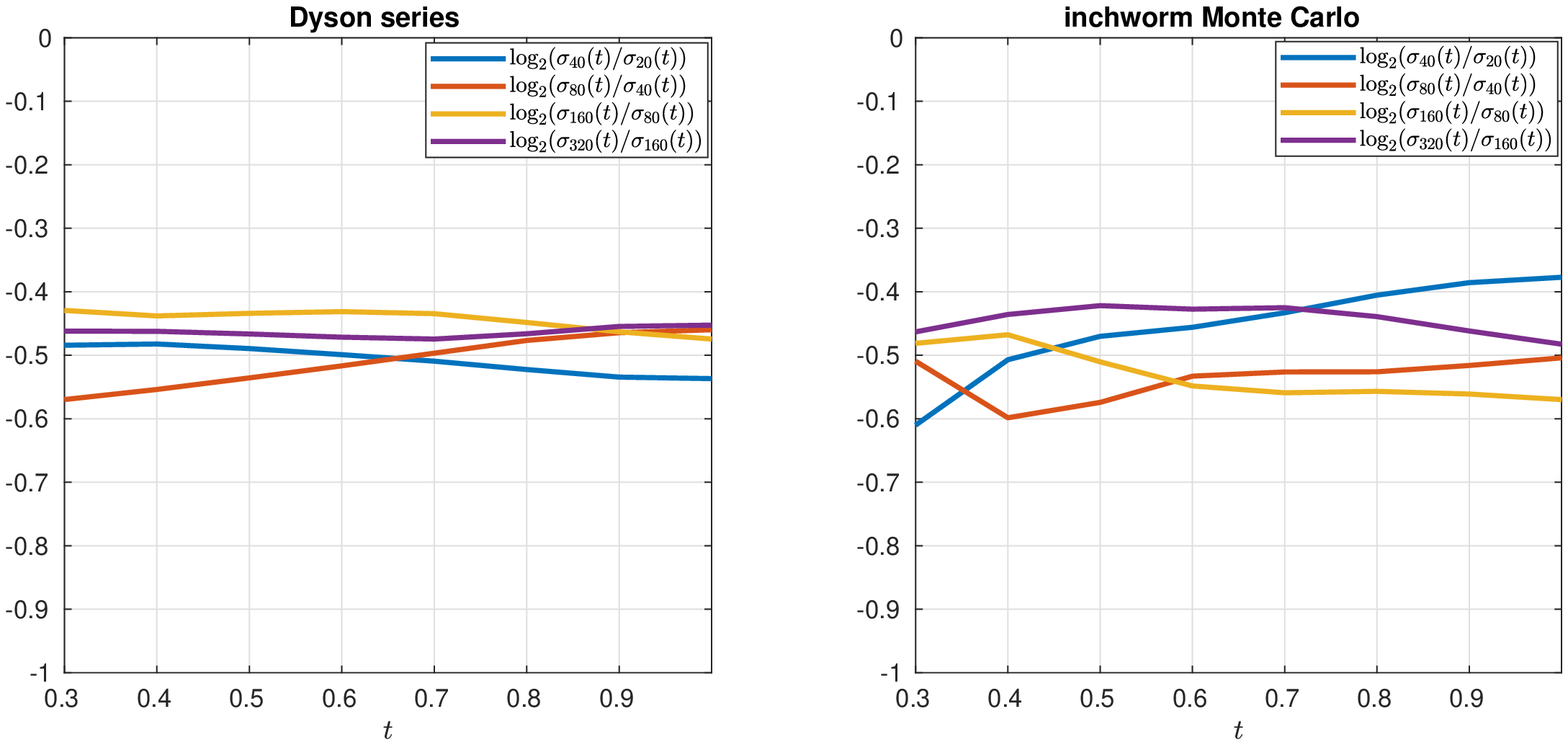}
    \caption{Evolution of the convergence rate of the standard derivation for the numerical solution.}
    \label{fig:conv order}
\end{figure}

\section{Conclusion}
\label{sec:conclusion}
We propose fast algorithms by reusing calculations of bath influence functionals to accelerate the summation of Dyson series and inchworm Monte Carlo method in the simulation of system-bath dynamics. For Dyson series, an integro-differential equation is derived, allowing us to solve the evolution of the observables using classic numerical schemes such as Runge-Kutta type methods. The idea of our fast algorithm is to make use of the invariance of the bath influence functionals so that any bath influence functional computed in the current time step can be reused in all the future time steps. Thanks to the linearity of the governing equation, the reuse algorithm for Dyson series can be implemented at a low memory cost. Such idea is then extended to the inchworm Monte Carlo method which computes the observables via the bivariate full propagator $G(\si,\sf)$, where the bath influence functionals calculated during the computation $G(\si,\sf)$ can be reused when computing $G(\si-\tau, \sf+\tau)$ for any $\tau > 0$. According to our complexity analysis, the computational cost is saved by a factor of $N$ with $N$ being the number of time steps, which makes our algorithms efficient for long time simulations. These theoretical results are further verified by numerical experiments.   

While we mainly focus on spin-boson model in this paper, our acceleration strategy can be also applied to general quantum systems interacting with the harmonic bath, where the value of two-point correlation $B(\tau_1,\tau_2)$ only relies on the time difference $\Delta \tau = |\tau_1|-|\tau_2|$ as in \eqref{def:B}. We also point out that for certain parameter settings of $B(\tau_1,\tau_2)$, a very small truncation at $\bar{M}=1$ or $\bar{M}=3$ may be sufficient for Dyson series and inchworm Monte Carlo method (see such examples in \cite[Section 7]{Cai2020}). In these cases, computational cost on the system integrand factor is comparable to that on the bath as shown in Table \ref{tab:U vs L time} and \ref{tab:U vs L time inchworm}. Therefore, including the system associated functional in the calculation reuse will be an interesting future direction. In addition, as the storage of bath influence functionals is the Achilles' heel of inchworm method in the current framework, further explorations into the memory cost reduction are also worth considering in future works.

\appendix

\section{Proof of statements}
\label{app:proof}
\subsection{Proof of \eqref{prop: B and Gs}}

\begin{proof}
\begin{itemize}
\item If $\si \le \sf <0$, we have 
\begin{displaymath}
 \begin{split}
&G_s^{(0)}(\si,\sf)^\dagger = \left(  \ee^{-\ii (\sf - \si) H_s} \right)^\dagger =  \ee^{-\ii (\si - \sf) H_s},\\
&\overline{B(\si,\sf)} = \overline{B^*(\sf - \si)} = B^*(\si -\sf)
 \end{split}
\end{displaymath} 
Since $0< -\sf \le -\si$ in this case, 
\begin{displaymath}
 \begin{split}
&G_s^{(0)}(-\sf,-\si) = \ee^{-\ii (\si - \sf) H_s} = G_s^{(0)}(\si,\sf)^\dagger, \\
& B(-\sf,-\si) = B^*(\si-\sf) = \overline{B(\si,\sf)}.
 \end{split}
\end{displaymath} 

\item If $0 < \si \le \sf$, we have $-\sf \le -\si < 0$ and thus 
\begin{displaymath}
\begin{split}
&G_s^{(0)}(-\sf,-\si) = \ee^{-\ii (\sf - \si) H_s}  = \left( \ee^{-\ii (\si - \sf) H_s} \right)^\dagger =  G_s^{(0)}(\si,\sf)^\dagger ,\\
& B(-\sf,-\si) = B^*(\sf - \si) = \overline{B^*(\si - \sf)} =  \overline{B(\si,\sf)}.
 \end{split}
\end{displaymath}

\item If $\si < 0 < \sf$, we have $-\sf < 0 < -\si$ and thus 
\begin{displaymath}
 \begin{split}
&G_s^{(0)}(-\sf,-\si) = \ee^{-\ii  \si H_s} O_s \ee^{-\ii  \sf H_s} = \left( \ee^{\ii  \sf H_s} O_s \ee^{\ii  \si H_s}  \right)^\dagger = G_s^{(0)}(\si,\sf)^\dagger,\\
&  B(-\sf,-\si) =  B^*( \si+\sf) =  \overline{B^*(-(\si + \sf))} = \overline{B(\si,\sf)}.
 \end{split}
\end{displaymath}
\end{itemize}

The above analysis excludes the special cases  $0 = \si \le \sf$ and $\si< \sf =0$ for which the statement for $B(\cdot,\cdot)$ is still true, while for $G_s^{(0)}(\cdot,\cdot)$ in general it is not due to the presence of $O_s$.

\begin{itemize}

\item  If $0 =\si < \sf$, we have $-\sf <0$ and   
\begin{displaymath}
 \begin{split}
&G_s^{(0)}(-\sf,-\si) = G_s^{(0)}(-\sf,0) = O_s \ee^{-\ii \sf H_s} = O_s G_s^{(0)}(\si,\sf)^\dagger,\\
& B(-\sf,-\si) = B(-\sf,0) = B^*(\sf)  = \overline{B^*(-\sf)} = \overline{B(\si,\sf)}.
 \end{split}
\end{displaymath}

\item  If $\si<\sf=0$, we have $-\si>0$ and 
\begin{displaymath}
 \begin{split}
&G_s^{(0)}(-\sf,-\si)O_s =  G_s^{(0)}(0,-\si)O_s = \ee^{-\ii \si H_s}O_s = G_s^{(0)}(\si,\sf)^\dagger,\\
& B(-\sf,-\si) =  B(0,-\si) = B^*(\si) = \overline{B^*(-\si)} =\overline{B(\si,\sf)}.
 \end{split}
\end{displaymath}

\item If $\si = \sf = 0$, 
\begin{displaymath}
 \begin{split}
&G_s^{(0)}(-\sf,-\si) = G_s^{(0)}(0,0) = I = G_s^{(0)}(\si,\sf)^\dagger,\\
&  B(-\sf,-\si) = B(0,0) = \frac{1}{\pi} \int^{\infty}_0 J(\omega)\dd \omega = \overline{B(\si,\sf)}.
 \end{split}
\end{displaymath}
\end{itemize}
\end{proof}

\subsection{Proof of \eqref{lemma:Lb} in Lemma \ref{lemma:dqmc}}
\begin{proof}
Define $s_{m+1} =t$ and $s'_{0} = -t$, we have   
  \begin{displaymath}
 \begin{split}
\Ls_b(-t,s_1,\cdots,s_m) = & \ \sum_{\mf{q}' \in \mQ(-t,\sb')} \prod_{(s'_j,s'_k) \in \mf{q}'} B(s'_j,s'_k) \\
= & \  \sum_{\mf{q}' \in \mQ(-t,\sb')} \prod_{(s'_j,s'_k) \in \mf{q}'} B(-s_{m+1-j},-s_{m+1-k}) \\
=& \ \sum_{\mf{q}' \in \mQ(-t,\sb')} \prod_{(s'_j,s'_k) \in \mf{q}'} \overline{B(s_{m+1-k},s_{m+1-j})} \\
\text{replace~}j'=m+1-k,k'=m+1-j  \Rightarrow \ =& \ \sum_{\mf{q}' \in \mQ(-t,-\sb)} \prod_{(s'_{m+1-k'},s'_{m+1-j'}) \in \mf{q}'} \overline{B(s_{j'},s_{k'})} \\
= & \ \sum_{\mf{q}' \in \mQ(-t,-\sb)} \prod_{(-s_{k'},-s_{j'}) \in \mf{q}'} \overline{B(s_{j'},s_{k'})} \\
= & \ \sum_{\mf{q} \in \mQ(\sb,t)} \prod_{(s_{j'},s_{k'}) \in \mf{q}} \overline{B(s_{j'},s_{k'})} = \overline{ \Ls_b(s_1,\cdots,s_m,t) }.
 \end{split}
\end{displaymath}

\end{proof}

\bibliographystyle{plain}
\bibliography{Inchworm}

\end{document}

%% file: images/fig_keldysh.tex
\begin{figure}[h]

\begin{tikzpicture}

\draw[thick] (-7,0) -- (-1,0);
\draw[thick] (1,0) -- (7,0);
\draw[thick] (-7,0.1)--(-7,-0.1);
\draw[thick] (-1,0.1)--(-1,-0.1);
\draw[thick] (7,0.1)--(7,-0.1);
\draw[thick] (1,0.1)--(1,-0.1);
\node[below left] at (-7,0) { $-t$};
\node[below right] at (-1,0) { $t$};
\node[below left] at (1,0) { $-t$};
\node[below right] at (7,0) { $t$};

\draw plot[only marks,mark =x, mark options={color=black, scale=1.5}]coordinates {(-4,0) (4,0)};

\draw plot[only marks,mark =*, mark options={color=black, scale=0.5}]coordinates { (3.5,0) (2.5,0) (5.5,0) (6.5,0) };

\draw plot[only marks,mark =*, mark options={color=black, scale=0.5}]coordinates { (-3.5,0) (-2.5,0) (-3,0) (-2,0) };

\draw (3.5,0) to[bend left=75] (5.5,0);
\draw (2.5,0) to[bend left=75] (6.5,0);

\draw (-3.5,0) to[bend left=75] (-2.5,0);
\draw (-3,0) to[bend left=75] (-2,0);

\node[below] at (-4,-0.5) {(a)};
\node[below] at (4,-0.5) {(b)};

\node[below] at (-4,0) {$0$};
\node[below] at (4,0) {$0$};

\node[below] at (-3.5,0) {$\tau_1$};
\node[below] at (-2.5,0) {$\tau_2$};
\node[below] at (-3,0) {$\tau'_1$};
\node[below] at (-2,0) {$\tau'_2$};

\node[below] at (3.5,0) {$\tau_1$};
\node[below] at (5.5,0) {$\tau_2$};
\node[below] at (2.5,0) {$\tau'_1$};
\node[below] at (6.5,0) {$\tau'_2$};

\end{tikzpicture}

\caption{Two cases of bath correlation invariance $B(\tau_1,\tau_2) = B(\tau'_1,\tau'_2)$.}

\label{fig:keldysh}
\end{figure}

%% file: images/property_L_revised.tex
\begin{figure}[h]

\begin{tikzpicture}

\draw[thick] (-5,0)--(-3,0);
\draw[thick] (-6,-1.5) -- (-2,-1.5);
\draw[thick] (-7,-3) -- (-1,-3);
\draw[thick,dashed]  (-4,0)--(-4,-3); 
\draw[thick,dashed] (-5,0)--(-5,-3);
\draw[thick,dashed] (-3,0)--(-3,-3);
\draw[thick,dashed]  (-6,-1.5)--(-6,-3);
\draw[thick,dashed]  (-2,-1.5)--(-2,-3);
\node[below left] at (-5,0) { $-h$};
\node[below right] at (-3,0) { $h$};
\node[below left] at (-6,-1.5) { $-2h$};
\node[below right] at (-2,-1.5) { $2h$};
\node[below left] at (-6.8,-3) { $-3h$};
\node[below right] at (-1.2,-3) { $3h$};
\draw[thick] (-5,-0.1)--(-5,0.1);
\draw[thick] (-3,-0.1)--(-3,0.1);
\draw[thick] (-6,-1.6)--(-6,-1.4);
\draw[thick] (-2,-1.6)--(-2,-1.4);
\draw[thick] (-7,-3.1)--(-7,-2.9);
\draw[thick] (-1,-3.1)--(-1,-2.9);
\draw plot[only marks,mark =*, mark options={color=black, scale=0.5}]coordinates { (-4.5,0) (-5.5,-1.5) (-3.5,-1.5) (-6.5,-3) (-2.5,-3) (-4.7,-3)  };

\draw[blue] (-4.5,0) to[bend left=60] (-3,0);

\draw[blue] (-5.5,-1.5) to[bend left=60] (-2,-1.5);

\draw[blue] (-6.5,-3)  to[bend left=50] (-1,-3);

\draw[red] (-3.5,-1.5) to[bend left=60] (-2,-1.5);

\draw[red] (-2.5,-3) to[bend left=50] (-1,-3);

\draw[green] (-4.7,-3) to[bend left=50] (-1,-3);

\end{tikzpicture}

\caption{Calculation reuse of $\Ls_b(\sb_i,t_i)$ for $m=1$.}

\label{fig:reuse}
\end{figure}

%% file: images/fig_mesh.tex

\resizebox{0.45\textwidth}{0.3\textheight}{
  
\begin{tikzpicture}
 
   \draw [->,thick] (-1,3)--(7,3);
     \draw [->,thick] (3,0)--(3,7);
        \draw [thick] (0,0)--(6,0);
   \node[right] at (7,3) {$\sf$};\node[above] at (3,7) {$\si$};
   
      \draw [dotted,thick] (0,0)--(0,3);
     \draw [dotted,thick] (1,1)--(1,3);
       \draw [dotted,thick] (2,2)--(2,3);
         \draw [dotted,thick] (4,4)--(3,4);
          \draw [dotted,thick] (5,5)--(3,5);
           \draw [dotted,thick] (6,6)--(3,6);


\draw (0,0)--(6,6);
\draw (1,0)--(3,2);
\draw (3,2)--(4,3);
\draw  (4,3)--(6,5);
\draw (2,0)--(3,1);
\draw (3,1)--(5,3);
\draw  (5,3)--(6,4);
\draw (3,0)--(6,3);
\draw (4,0)--(6,2);
\draw (5,0)--(6,1);

\draw (1,0)--(1,1)--(6,1);
\draw (2,0)--(2,2)--(6,2);
\draw (3,0)--(3,3)--(6,3);
\draw (4,0)--(4,4)--(6,4);
\draw (5,0)--(5,5)--(6,5);
\draw[thick]  (6,0)--(6,6);

 \draw[->,line width=0.7mm]  (1,1)--(1,0);  
   \draw[->,line width=0.7mm]  (2,2)--(2,0);  
    \draw[->,line width=0.7mm]  (3,3)--(3,0);  
     \draw[->,line width=0.7mm]  (4,4)--(4,0);  
      \draw[->,line width=0.7mm]  (5,5)--(5,0);  
       \draw[->,line width=0.7mm]  (6,6)--(6,0);  
   

   
 

\draw plot[only marks,mark=*, mark options={color=black, scale=1}] coordinates {(0,0) (1,1) (2,2)   (4,4)  (5,5) (6,6) };

\draw plot[only marks,mark=*, mark options={color=black, scale=1}] coordinates {(1,0) (2,0) (2,1)  (4,0) (4,1) (4,2)  (5,0) (5,1) (5,2) (5,4) (6,0) (6,1) (6,2) (6,4) (6,5) };

\draw plot[only marks,mark=*, mark options={color=black, scale=1}] coordinates {(3,0) (3,1) (3,2)(4,3) (5,3) (6,3)  (3,3) };

 \node[above] at (0,3) {$-0.3$}; 
  \node[above] at (1,3) {$-0.2$};
   \node[above] at (2,3) {$-0.1$};  
 \node[above left] at (3,3) {$0$};
   \node[left] at (3,4) {$0.1$};
  \node[left] at (3,5) {$0.2$};
    \node[left] at (3,6) {$0.3$};

   \draw [thick,red] (4.8,0.8)--(5.2,0.8)--(5.2,1.2)--(4.8,1.2)--(4.8,0.8);
    
      \draw [thick,blue](4.8,1.8)--(5.2,1.8)--(5.2,2.2)--(4.8,2.2)--(4.8,1.8);
  \end{tikzpicture}
  }


%% file: images/fig_mesh_2.tex
\begin{figure}[h]
\centering

\def\h{0.8}

\begin{tikzpicture}
 
   \draw [->,thick] ({-10*\h},0)--({-2*\h},0);
     \draw [->,thick] ({-6*\h},{-3*\h})--({-6*\h},{4*\h});
        \draw [thick] ({-9*\h},{-3*\h})--({-3*\h},{-3*\h});
   \node[right] at ({-2*\h},0) {$k$};\node[above] at ({-6*\h},{4*\h}) {$j$};
   \node[below] at ({-6*\h},{-4*\h}) {(a)};
   
      \draw [dotted,thick] ({-9*\h},{-3*\h})--({-9*\h},0);
     \draw [dotted,thick] ({-8*\h},{-2*\h})--({-8*\h},0);
       \draw [dotted,thick] ({-7*\h},{-\h})--({-7*\h},0);
         \draw [dotted,thick] ({-5*\h},{\h})--({-6*\h},{\h});
          \draw [dotted,thick] ({-4*\h},{2*\h})--({-6*\h},{2*\h});
           \draw [dotted,thick] ({-3*\h},{3*\h})--({-6*\h},{3*\h});

     \fill [green,opacity=.5] ({-6*\h},{-3*\h})--({-3*\h},{-3*\h})--({-6*\h},0);

\draw ({-9*\h},{-3*\h})--({-3*\h},{3*\h});
\draw [-,red,line width = 1mm] ({-8*\h},{-3*\h})--({-6*\h},{-\h});
\draw [-,red,line width = 1mm] ({-5*\h},0)--({-3*\h},{2*\h});
\draw [-,blue,line width = 1mm] ({-7*\h},{-3*\h})--({-6*\h},{-2*\h});
\draw [-,blue,line width = 1mm] ({-4*\h},0)--({-3*\h},{\h});

\draw[dotted,thick]  ({-8*\h},{-3*\h})--({-8*\h},{-2*\h})--({-3*\h},{-2*\h});
\draw[dotted,thick]  ({-7*\h},{-3*\h})--({-7*\h},{-\h})--({-3*\h},{-\h});
\draw[dotted,thick]  ({-5*\h},{-3*\h})--({-5*\h},{\h})--({-3*\h},{\h});
\draw[dotted,thick]  ({-4*\h},{-3*\h})--({-4*\h},{2*\h})--({-3*\h},{2*\h});
\draw ({-3*\h},{-3*\h})--({-3*\h},{3*\h});

   \draw [-,red,line width = 1mm] ({-6*\h},{-\h})--({-5*\h},{-\h})--({-5*\h},0);
     \draw [-,blue,line width = 1mm] ({-6*\h},{-2*\h})--({-4*\h},{-2*\h})--({-4*\h},0);
       \draw [-,black,line width = 1mm] ({-6*\h},{-3*\h})--({-3*\h},{-3*\h})--({-3*\h},0);

   
 

\draw plot[only marks,mark=*, mark options={color=red, scale=1}] coordinates {({-9*\h},{-3*\h}) ({-8*\h},{-2*\h}) ({-7*\h},{-\h})   ({-5*\h},{\h})  ({-4*\h},{2*\h}) ({-3*\h},{3*\h}) };

\draw plot[only marks,mark=*, mark options={color=black, scale=1}] coordinates {({-8*\h},{-3*\h}) ({-7*\h},{-3*\h})({-5*\h},{-3*\h})({-4*\h},{-3*\h})({-3*\h},{-3*\h})({-7*\h},{-2*\h})({-5*\h},{-2*\h})({-4*\h},{-2*\h})({-3*\h},{-2*\h}) ({-5*\h},{-\h})({-4*\h},{-\h})({-3*\h},{-\h}) ({-4*\h},{\h})({-3*\h},{\h})({-3*\h},{2*\h}) };

\draw plot[only marks,mark=*, mark options={color=blue, scale=1}] coordinates {({-6*\h},{-3*\h}) ({-6*\h},{-2*\h}) ({-6*\h},{-\h})({-6*\h},0) ({-5*\h},0) ({-4*\h},0)  ({-3*\h},0) };

 \node[above] at ({-9*\h},0) {\small $-3$}; 
  \node[above] at ({-8*\h},0) {\small $-2$};
   \node[above] at ({-7*\h},0) {\small $-1$};  
 \node[above left] at ({-6*\h},0) {\small $0$};
   \node[left] at ({-6*\h},{\h}) {\small $1$};
  \node[left] at ({-6*\h},{2*\h}) {\small $2$};
    \node[left] at ({-6*\h},{3*\h}) {\small $3$};


    \draw [->,thick] ({\h},0)--({9*\h},0);
     \draw [->,thick] ({5*\h},{-3*\h})--({5*\h},{4*\h});
        \draw [thick] ({2*\h},{-3*\h})--({8*\h},{-3*\h});
         \draw [thick] ({2*\h},{-3*\h})--({8*\h},{3*\h});
            \draw [thick] ({8*\h},{-3*\h})--({8*\h},{3*\h});
    \draw [dotted,thick] ({3*\h},{-3*\h})--({3*\h},{-2*\h})--({8*\h},{-2*\h});
    \draw [dotted,thick] ({4*\h},{-3*\h})--({4*\h},{-\h})--({8*\h},{-\h});
    \draw [dotted,thick]({6*\h},{-3*\h})--({6*\h},{\h})--({8*\h},{\h});
    \draw [dotted,thick]({7*\h},{-3*\h})--({7*\h},{2*\h})--({8*\h},{2*\h});
        
   \node[right] at ({9*\h},0) {$k$};\node[above] at ({5*\h},{4*\h}) {$p$};
   
      \draw [dotted,thick] ({2*\h},{-3*\h})--({2*\h},0);
     \draw [dotted,thick] ({3*\h},{-2*\h})--({3*\h},0);
       \draw [dotted,thick] ({4*\h},{-\h})--({4*\h},0);
         \draw [dotted,thick] ({6*\h},{\h})--({5*\h},{\h});
          \draw [dotted,thick] ({7*\h},{2*\h})--({5*\h},{2*\h});
           \draw [dotted,thick] ({8*\h},{3*\h})--({5*\h},{3*\h});

      \draw [->,blue,line width = 0.75mm] ({5*\h},-2*\h)--({6*\h},{-3*\h});
            \draw [->,red,line width = 0.75mm] ({5*\h},-\h)--({7*\h},{-3*\h});
            
             \draw [->,blue,line width = 0.75mm] ({6*\h},-\h)--({8*\h},{-3*\h});
                    \draw [->,blue,line width = 0.75mm] ({7*\h},-\h)--({8*\h},{-2*\h});

   \draw plot[only marks,mark=*, mark options={color=darkgreen, scale=1}] coordinates {({5*\h},{-3*\h}) ({5*\h},{-2*\h}) ({5*\h},{-\h}) ({6*\h},{-3*\h}) ({6*\h},{-2*\h})({6*\h},{-\h})  ({7*\h},{-3*\h}) ({7*\h},{-2*\h}) ({7*\h},{-\h})  ({8*\h},{-3*\h})  ({8*\h},{-2*\h}) ({8*\h},{-\h})    };

      \draw plot[only marks,mark=x, mark options={color=darkgreen, scale=1.5}] coordinates { ({6*\h},0)  ({7*\h},0) ({7*\h},{\h})  ({8*\h},{2*\h})  ({8*\h},{\h}) ({8*\h},0)    };

    \node[above] at ({2*\h},0) {\small $-3$}; 
  \node[above] at ({3*\h},0) {\small $-2$};
   \node[above] at ({4*\h},0) {\small $-1$};  
 \node[above left] at ({5*\h},0) {\small $0$};
   \node[left] at ({5*\h},{\h}) {\small $1$};
  \node[left] at ({5*\h},{2*\h}) {\small $2$};
    \node[left] at ({5*\h},{3*\h}) {\small $3$};
      \node[below] at ({5*\h},{-4*\h}) {(b)};

  \end{tikzpicture}

  \caption{An example for $N = 3$ (left: time evolution of inchworm Monte Carlo method, right: sampling and reuse strategy).}
 \label{fig:reuse_order}
\end{figure}

%% file: images/property_L_inchworm.tex
\begin{figure}[h]

\begin{tikzpicture}




\draw[thick] (-2,0)--(0,0);
\draw[thick] (-4,-2) -- (2,-2);
\draw[thick] (-6,-4) -- (4,-4);
\draw[thick,dashed]  (-2,0)--(-2,-4); 
\draw[thick,dashed]  (0,0)--(0,-4); 
\draw[thick,dashed]  (-4,-2)--(-4,-4); 
\draw[thick,dashed]  (2,-2)--(2,-4); 

\node[below right] at (0,0) { $0$};
\node[below left] at (-2,0) { $-h$};
\node[below left] at (-4,-2) { $-2h$};
\node[below right] at (2,-2) { $h$};
\node[below left] at (-6,-4) { $-3h$};
\node[below right] at (4,-4) { $2h$};

\draw[thick] (-2,-0.1)--(-2,0.1);
\draw[thick] (0,-0.1)--(0,0.1);
\draw[thick] (-2,-1.9)--(-2,-2.1);
\draw[thick] (0,-1.9)--(0,-2.1);
\draw[thick] (-2,-3.9)--(-2,-4.1);
\draw[thick] (0,-3.9)--(0,-4.1);
\draw[thick] (-4,-1.9)--(-4,-2.1);
\draw[thick] (2,-1.9)--(2,-2.1);
\draw[thick] (-4,-3.9)--(-4,-4.1);
\draw[thick] (2,-3.9)--(2,-4.1);
\draw[thick] (-6,-3.9)--(-6,-4.1);
\draw[thick] (4,-3.9)--(4,-4.1);

\draw plot[only marks,mark =*, mark options={color=black, scale=0.7}]coordinates {  (-1.33,0) (-1,0)   (-3.33,-2) (-3,-2)   (-5.33,-4) (-5,-4)  (-1.33,-2) (-0.67,-2)   (-3.33,-4) (-2.67,-4)  (-1.33,-4) (-0.67,-4)   };

\draw plot[only marks,mark =*, mark options={color=red, scale=0.7}]coordinates { (-1.67,0) (-3.67,-2)  (-5.67,-4)  (-2.33,-2)   (-4.33,-4)     (-4.67,-4)    };

\draw[blue] (-1.67,0) to[bend left = 60]  (-1,0);
\draw[blue] (-1.33,0) to[bend left] (0,0);
\draw[blue] (-3.67,-2) to[bend left=60]  (-3,-2);
\draw[blue] (-3.33,-2) to[bend left] (2,-2);
\draw[blue] (-5.67,-4) to[bend left=60]  (-5,-4);
\draw[blue] (-5.33,-4) to[bend left] (4,-4);
\draw[red]  (-2.33,-2)to[bend left] (-0.67,-2);
\draw[red]  (-1.33,-2)to[bend left=25] (2,-2);
\draw[red]  (-4.33,-4)to[bend left] (-2.67,-4);
\draw[red]  (-3.33,-4)to[bend left=25] (4,-4);
\draw[green]  (-4.67,-4)to[bend left]  (-0.67,-4);
\draw[green] (-1.33,-4)to[bend left=25] (4,-4);
\end{tikzpicture}

\caption{Calculation reuse of $\Ls_b^c(s_1,s_2,s_3,t_k)$ along the path $G_{-1,0} \rightarrow G_{-2,1} \rightarrow G_{-3,2}$.}

\label{fig:reuse inchworm}
\end{figure}